\documentclass[reqno,12pt]{article}
\usepackage{amsmath,amsthm}
\usepackage{bm}
\usepackage{a4wide,amssymb}
\usepackage{graphicx,xcolor}
\usepackage{bbm}
\newcommand{\R}{{\mathbb R}}         
\newcommand{\bR}{{\mathbb R}}         
 
\newcommand{\E}{{\mathbb E}}
\renewcommand{\P}{{\mathbb P}}
\newcommand{\cF}{{\mathcal F}}
\newcommand{\cA}{{\mathcal A}}
\newcommand{\cU}{{\mathcal U}}
\newcommand{\cL}{{\mathcal L}}
\newcommand{\cK}{{\mathcal K}}
\newcommand{\cO}{{\mathcal O}}
\newcommand{\cD}{{\mathcal D}}
\newcommand{\cJ}{{\mathcal J}}
\newcommand{\cN}{{\mathcal N}}
\newcommand{\N}{{\mathbb N}}
\newcommand{\bN}{{\mathbb N}}
\newcommand{\Nn}{{\mathcal N}}

\newcommand{\cH}{{\mathcal H}}
\newcommand{\tr}{\mbox{Tr}}
\newcommand{\wt}{\widetilde}

\newcommand{\eps}{\varepsilon}

\newcommand{\ph}{\varphi}

\def\f{{\varphi}}

\def\U{{\mathcal U}}

\newtheorem{thm}{Theorem}[section]
\newtheorem{lem}[thm]{Lemma}
\newtheorem{prop}[thm]{Proposition}
\newtheorem{cor}[thm]{Corollary}

\def\be{\begin{equation}}
\def\ee{\end{equation}}
\def\bee{\begin{equation*}}
\def\eee{\end{equation*}}
\def\bsplit{\begin{split}}
\def\esplit{\end{split}}
\def\bea{\begin{eqnarray}}
\def\eea{\end{eqnarray}}

\def\nn{\nonumber}

\def\ol{\overline}

\begin{document}

\title{Multivariate Central Limit Theorem \\ in Quantum Dynamics}

\author{Simon Buchholz, Chiara Saffirio\thanks{Supported by ERC Grant MAQD 240518 } \, and Benjamin Schlein\thanks{Partially supported by ERC Grant MAQD 240518}\\ \\ Institute of Applied Mathematics, University of Bonn\\ Endenicher Allee 60, 53115 Bonn, Germany}

\maketitle

\begin{abstract}
We consider the time evolution of $N$ bosons in the mean field regime for factorized initial data. In the limit of large $N$, the many body evolution can be approximated by the non-linear Hartree equation. In this paper we are interested in the fluctuations around the Hartree dynamics.  
We choose $k$ self-adjoint one-particle operators $O_1, \dots , O_k$ on $L^2 (\bR^3)$, and we average their action over the $N$-particles. We show that, for every fixed $t \in \bR$, expectations of products of functions of the averaged observables approach, as $N \to \infty$, expectations with respect to a complex Gaussian measure, whose covariance matrix can be expressed in terms of a Bogoliubov transformation describing the dynamics of quantum fluctuations around the mean field Hartree evolution. If the operators $O_1, \dots , O_k$ commute, the Gaussian measure is real and positive, and we recover a ``classical'' multivariate central limit theorem. All our results give explicit bounds on the rate of the convergence (we obtain therefore Berry-Ess{\'e}en type central limit theorems).  
\end{abstract}


\section{Introduction}
\setcounter{equation}{0}

We consider a system of $N$ identical particles in three dimensions, described by a normalized wave function $\psi_N\in L^2(\R^{3N})$. We are interested in particles obeying bosonic statistics, meaning that $\psi_N$ is symmetric with respect to any permutation of the $N$ particles, in the sense that  
\be\label{eq:sym}
\psi_N (x_{\pi_1},\dots,x_{\pi_N})=\psi_N(x_1,\dots,x_N) \, 
\ee
for any permutation $\pi \in S_N$. We denote by $L^2_s (\bR^{3N})$ the subspace of $L^2 (\bR^{3N})$ consisting of permutation symmetric wave functions, satisfying (\ref{eq:sym}). 

\medskip

{\it Mean field regime.} We will focus on the mean field regime of many body quantum mechanics, which is characterized by the fact that every particle experiences a very large number of very weak collisions, so that the total force is comparable with the inertia of the particles. To study the mean field regime, we define the Hamilton operator 
\be\label{eq:ham}
H_N = \sum_{i=1}^N-\Delta_{x_j} + \frac{1}{N} \sum_{i<j}^N V(x_i-x_j)\;.
\ee
and consider the evolution generated by $H_N$, which is governed by the $N$ particle Schr\"odinger equation
\be\label{eq:SE}
i\partial_t \psi_{N,t} = H_N \psi_{N,t}\; 
\ee 
whose solution can be written as $\psi_{N,t} = e^{-iH_N t} \psi_{N,0}$, where $\psi_{N,0}$ denotes the initial wave function, at time $t=0$. In (\ref{eq:ham}), $V(x_i -x_j)$ describes the interaction between particle $i$ and particle $j$; we will assume the potential $V$ to satisfy the operator inequality \begin{equation}\label{eq:bdV2} V^2 (x) \leq D (1-\Delta),\end{equation} for some constant $D>0$. In particular, this inequality is satisfied for the physically relevant example of a Coulomb potential $V(x) = -1/|x|$. In order to simplify a bit the notation, we do not include in (\ref{eq:ham}) external potentials; nevertheless our results and our techniques remain valid if $-\Delta_{x_j}$ is replaced by $-\Delta_{x_j} + V_{\text{ext}} (x_j)$, under very general conditions on $V_{\text{ext}}$. 

\medskip

{\it Evolution of factorized initial data.} If the Hamiltonian (\ref{eq:ham}) is restricted to a finite domain with volume of order one (either by imposing boundary conditions or by adding a trapping external potential), the ground state is known to exhibit complete condensation, meaning that, in an appropriate sense $\psi_N \simeq \ph^{\otimes N}$ for a $\ph \in L^2 (\bR^3)$ (the one-particle orbital $\ph$ is the minimizer of the Hartree energy, which takes into account the trapping potential). For this reason, one is typically interested in the time-evolution of factorized (or at least approximately factorized) initial data. It turns out that, if at time $t=0$, $\psi_{N} \simeq \ph^{\otimes N}$, then the solution of the many body Schr\"odinger equation (\ref{eq:SE}) remains of the form $\psi_{N,t} \simeq \ph_t^{\otimes N}$, where $\ph_t$ solves the nonlinear time-dependent Hartree equation 
\begin{equation}\label{eq:hartree0} i\partial_t \ph_t = -\Delta \ph_t + (V*|\ph_t|^2) \ph_t \end{equation}
with the initial data $\ph_{t=0} = \ph$. In other words, complete condensation is preserved by the time-evolution and the dynamics of the condensate wave function is governed by the Hartree equation (\ref{eq:hartree0}). 

\medskip

{\it Reduced density matrices.} In order to obtain a precise mathematical statement about the convergence towards the Hartree equation (\ref{eq:hartree0}), we introduce the notion of reduced density matrices. For $k = 1, \dots , N$, we define the $k$-particle reduced density matrix $\gamma^{(k)}_{N,t}$ associated with the solution $\psi_{N,t}$ of the Schr\"odinger equation by taking the partial trace of the orthogonal projection $|\psi_{N,t} \rangle \langle \psi_{N,t}|$ over the last $N-k$ particles:
\[ \gamma^{(k)}_{N,t} = \tr_{k+1, k+2, \dots , N} \, |\psi_{N,t} \rangle \langle \psi_{N,t}|\;. \]
In other words, we define $\gamma^{(k)}_{N,t}$ as the non-negative trace-class operator on $L^2 (\bR^{3k})$ with integral kernel 
\[ \begin{split} \gamma^{(k)}_{N,t} (x_1, &\dots , x_k; x'_1 , \dots , x'_k) \\ &= \int dx_{k+1} \dots dx_{N} \, \psi_{N,t} (x_1, \dots , x_k , x_{k+1}, \dots , x_N) \overline{\psi}_{N,t} (x'_1, \dots , x'_k, x_{k+1} , \dots , x_N) \, . \end{split} \]
Observe that knowledge of $\gamma^{(k)}_{N,t}$ is sufficient to compute the expectation of any $k$-particle observable. In fact, if $J$ is a self-adjoint operator on $L^2 (\bR^{3k})$ and $J\otimes 1$ is the self-adjoint operator on $L^2 (\bR^{3N})$ which acts as $J$ on the first $k$ particles, and as the identity on the other $(N-k)$ particles, we have
\[ \langle \psi_{N,t} ,  \left(J \otimes 1 \right) \psi_{N,t} \rangle = \tr \, |\psi_{N,t} \rangle \langle \psi_{N,t} | \, (J \otimes 1) = \tr \, \gamma^{(k)}_{N,t} J \, . \]

\medskip

{\it Convergence towards Hartree dynamics.} It turns out that the language of the reduced densities is the appropriate language to understand the convergence towards the Hartree dynamics (\ref{eq:hartree0}). Consider the initial data $\psi_{N} = \ph^{\otimes N}$ (but the following result can be extended to more general initial data), and assume that the interaction potential $V$ satisfies (\ref{eq:bdV2}). Let $\psi_{N,t} = e^{-i H_N t} \psi_N$ and denote by $\gamma^{(k)}_{N,t}$ the $k$-particles reduced density associated with $\psi_{N,t}$. Then, for every $k \in \bN$, there exist constants $C_k,c_k > 0$ such that
\begin{equation}\label{eq:conv-mar} \tr \,\left| \gamma^{(k)}_{N,t} - |\ph_t \rangle \langle \ph_t|^{\otimes k} \right| \leq \frac{C_k  \exp \, (c_k |t|)}{N} \end{equation}
for all $t \in \bR$ and $N$ large enough. In particular, this implies convergence of the expectation of arbitrary observables depending only on a finite number of particles. If $J$ is a self-adjoint operator on $L^2 (\bR^{3k})$, then
\[ \Big| \langle \psi_{N,t} , (J \otimes 1) \psi_{N,t} \rangle - \langle \ph_t^{\otimes k}, J \ph_t^{\otimes k} \rangle \Big| \leq \frac{C_k \exp (c_k |t|)}{N}\, . \]
The first mathematically rigorous works which established the convergence of the many body dynamics towards the Hartree evolution were based on the study of the evolution of the reduced densities $\gamma^{(k)}_{N,t}$ as described by the BBGKY hierarchy of equations; see \cite{Sp,EY}. Later, the BBGKY approach was also extended to the so called Gross-Pitaevskii regime, in which the interaction potential $V$ depends on $N$, varying on a length-scale of order $N^{-1}$ and converging towards a delta-function in the limit of large $N$ (in this case, the system cannot be interpreted as describing a mean-field regime); see \cite{ESY1,ESY2,ESY3}. All these results do not give a bound on the rate of the convergence towards the Hartree dynamics. A different approach, leading to the quantitative estimate (\ref{eq:conv-mar}), was later developed in \cite{RS} and later extended in \cite{CLS} and to the Gross-Pitaevskii regime in \cite{BDS}. This approach, which follows ideas originally introduced in \cite{H,GV}, is based on a representation of the system on the bosonic Fock space and on the study of the time evolution of initial coherent states. Since this method will play a central role in our paper, we discuss its main ideas in Section \ref{sec:fock}. Notice that, recently, different approaches to obtain a mathematical understanding of the time evolution in the mean field regime have been developed in \cite{KP} and in \cite{FKS}, where the convergence towards the Hartree dynamics is formulated as a Egorov-type theorem.

\medskip

{\it A law of large numbers.} It is possible to translate the convergence (\ref{eq:conv-mar}) in a more probabilistic language. Let $O$ be a bounded self-adjoint operator on $L^2 (\bR^3)$ and denote by $O^{(j)} = 1 \otimes \dots \otimes O \otimes \dots \otimes 1$ the operator acting as $O$ on the $j$-th particle and as the identity on the other $(N-1)$ particles. Given a wave function $\psi \in L^2_s (\bR^{3N})$, one can think of each $O^{(j)}$ as a random variable, whose probability distribution is determined by $\psi$ through the spectral theorem. The probability that $O^{(j)}$ assumes values in $A \subset \bR$ is given by 
\[ \P_\psi (O^{(j)} \in A) = \langle \psi, \chi_A (O^{(j)}) \psi \rangle \]
where $\chi_A$ is the characteristic function of the set $A$. Consider the factorized wave function 
$\psi_{N,0} = \ph^{\otimes N}$. With respect to $\psi_{N,0}$, the random variables $O^{(j)}$ are independent and identically distributed. Consider now the evolved wave function $\psi_{N,t} = e^{-iH_N t} \psi_{N,0}$, where $H_N$ is the mean field $N$-particle Hamiltonian (\ref{eq:ham}). With respect to $\psi_{N,t}$ the random variables $O^{(j)}$ are not independent. Nevertheless, (\ref{eq:conv-mar}) implies a law of large numbers, in the sense that, for every $\delta > 0$, 
\begin{equation}\label{eq:lln} \lim_{N \to \infty} \P_{\psi_{N,t}} \left( \left| \frac{1}{N} \sum_{j=1}^N O^{(j)} - \langle \ph_t , O \ph_t \rangle \right| > \delta \right)  = 0\;. \end{equation}
In fact
\[ \P_{\psi_{N,t}} \left( \left| \frac{1}{N} \sum_{j=1}^N O^{(j)} - \langle \ph_t , O \ph_t \rangle \right| > \delta \right) =  \P_{\psi_{N,t}} \left( \left| \frac{1}{N\delta} \sum_{j=1}^N \wt{O}^{(j)} \right| \geq 1 \right) \]
where $\wt{O}^{(j)} = O^{(j)} - \langle \ph_t , O \ph_t\rangle$. Markov's inequality therefore implies that 
\[  \begin{split} \P_{\psi_{N,t}} \left( \left| \frac{1}{N} \sum_{j=1}^N O^{(j)} - \langle \ph_t , O \ph_t \rangle \right| > \delta \right) &\leq \frac{1}{N^2 \delta^2} \,  \E_{\psi_{N,t}} \left| \sum_{j=1}^N \wt{O}^{(j)} \right|^2 
\\ &\leq \frac{1}{\delta^{2}} \tr \, \gamma_{N,t}^{(2)} \, (\wt{O} \otimes \wt{O}) + \frac{1}{N\delta^{2}} \,  \tr \gamma_{N,t}^{(1)} \, \wt{O}^2 \end{split} \]
with $\wt{O} = O - \langle \ph_t , O \ph_t \rangle$. On the one hand, we have
\[ \frac{1}{\delta^2 N} \tr \, \gamma^{(1)}_{N,t} \, \wt{O}^2 \leq \frac{\| \wt{O} \|^2}{\delta^2 N} \to 0 \]
as $N \to \infty$. On the other hand
\[ \tr \, \gamma^{(2)}_{N,t} \, (\wt{O} \otimes \wt{O}) \to \tr \, |\ph_t \rangle \langle \ph_t |^{\otimes 2} \, (\wt{O} \otimes \wt{O}) = \langle \ph_t, \wt{O} \ph_t \rangle^2 = 0 \] 
as $N \to \infty$. This implies (\ref{eq:lln}). 

\medskip

{\it A central limit theorem.} After establishing the law of large numbers (\ref{eq:lln}), one can investigate the fluctuations around the Hartree dynamics. In \cite{BKS} it was proven that, under some regularity conditions on the self-adjoint operator $O$ on $L^2 (\bR^3)$, the appropriately rescaled random variable
\begin{equation}\label{eq:cOt} \cO_t := \frac{1}{\sqrt{N}} \sum_{j=1}^N \left( O^{(j)} - \langle \ph_t, O \ph_t \rangle \right) \end{equation}
converges in distribution, as $N \to \infty$, to  a centered Gaussian random variable with variance
\[ \sigma_t^2 =  \| U(t;0) O \ph_t + \overline{V(t;0) O \ph_t} \|^2 - |\langle \ph , U(t;0) O \ph_t + \overline{V(t;0) O \ph_t} \rangle|^2 \]
where $U (t;s), V(t;s): L^2 (\bR^3) \to L^2 (\bR^3)$ are families of linear maps, defining a so called Bogoliubov transformations, which emerge naturally in the study of the time evolution of coherent states and describe fluctuations around the mean field Hartree limit. We will give the precise definition of the maps $U(t;s), V(t;s)$ (and of the associated Bogoliubov transformations $\Theta (t;s)$) in Section \ref{sec:fock}. Observe that with respect to the measure induced by the factorized wave function $\ph_t^{\otimes N}$, the random variable (\ref{eq:cOt}) converges to a centered Gaussian, with the variance $\wt{\sigma}^2_t = \langle \ph_t , O^2 \ph_t\rangle - \langle \ph_t, O \ph_t \rangle^2$. This means that, while the correlations among the particles in $\psi_{N,t}$ are sufficiently weak for a central limit theorem to hold true, they are strong enough to change the variance of the limiting Gaussian. A different approach to study fluctuations around the mean field dynamics has been explored in \cite{GMM,GMM2,C} and, more recently, in \cite{LNS} (similar results have been obtained in the static time-independent setting, in \cite{GS,LNSS}; in this case, one considers the excitation spectrum of the Hamiltonian (\ref{eq:ham}), after imposing an external confining potential). In different settings, quantum central limit theorems have been previously established in \cite{CH,HL,GVV,JPP,S,Ha,Ku,CE}. 

\medskip

{\it Multivariate central limit theorem.} A natural question emerging from the result of \cite{BKS} is whether one can also establish a multivariate version of the central limit theorem. Let $k \in \bN$ and let $O_1, \dots , O_k$ be bounded operators on $L^2 (\bR^3)$. For $j=1, \dots , k$, we define 
\begin{equation}\label{eq:cOjt} \cO_{j,t} = \frac{1}{\sqrt{N}} \sum_{i=1}^N \left( O_j^{(i)} - \langle \ph_t, O_j \ph_t \rangle \right)\;. \end{equation}
At this point we observe that there is an important difference with respect to standard probability theory. Unless the operators $O_1, \dots , O_k$ commute among each others, they cannot be measured simultaneously. For this reason it does not make sense to ask about the joint probability distribution of the random variables $\cO_{1,t}, \dots , \cO_{k,t}$. One can still ask about expectations of products of functions of these observables. In contrast with classical probability, however, these expectations do not need to be real (because the product of self-adjoint operator does not need to be self-adjoint). Our main result is the following theorem, which shows that, expectations of products of functions of $\cO_{1,t}, \dots , \cO_{k,t}$ can be computed integrating the functions $f_1, \dots, f_k$ against a complex-valued Gaussian density, with covariance matrix expressed in terms of the Bogoliubov transformation $\Theta (t;s)$ appearing in the central limit theorem shown in \cite{BKS}.

\begin{thm}\label{thm:multi-CLT}
Let $V$ satisfy (\ref{eq:bdV2}). Let $\ph \in H^2 (\bR^3)$ and let $\psi_{N,t}$ denote the solution of the Schr\"odinger equation (\ref{eq:SE}), with the initial data $\psi_{N,0} = \ph^{\otimes N}$. Let $O_1,\dots, O_k$ be self-adjoint operators on $L^2(\R^3)$, such that $\| \partial^{\alpha} O_j (1- \Delta)^{-|\alpha|/2} \| < \infty$ for every multi-index $\alpha \in \bN^3$ with $ |\alpha| \leq 2$ and every $j = 1, \dots , k$, and define $\cO_{j,t}$ as in (\ref{eq:cOjt}). Let $f_1, \dots , f_k \in L^1 (\bR)$ with $\widehat{f}_j \in L^1 (\bR, (1+ \tau^{10}) d\tau)$. For any $t \in \bR$, we define the complex $k \times k$ covariance matrix $\Sigma (t) = (\Sigma_{ij} (t))_{1\leq i,j \leq k}$ by
\[ \Sigma_{ij} (t)= \langle g_{i,t} , g_{j,t} \rangle - \langle g_{i,t} , \ph \rangle \langle \ph , g_{j,t} \rangle \]
for all $i \leq j$ and by $\Sigma_{ij}  (t) = \Sigma_{ji} (t)$ for all $i > j$. Here $g_{1,t}, \dots , g_{k,t} \in L^2 (\bR^3)$ are given by
\[ g_{j,t} = U(t;0) O_j \ph_t + \overline{V(t;0) O_j \ph_t} \,  \]
where $U(t;0), V(t;0) : L^2 (\bR^3) \to L^2 (\bR^3)$ are linear maps defined in Proposition \ref{prop:bog} below: they are the block-components of the Bogoliubov transformation $\Theta (t;0): L^2 (\bR^3) \oplus L^2 (\bR^3) \to L^2 (\bR^3) \oplus L^2 (\bR^3)$ describing the action of the limiting fluctuation dynamics $\cU_\infty$ defined in Proposition \ref{prop:GV}. 

The real part $\text{Re } \Sigma (t) = (\Sigma (t) + \Sigma^* (t))/2$ is a non-negative symmetric matrix. We assume $\text{Re } \Sigma (t)$ to be strictly positive. Then, there exist constants 
$C,K> 0$ such that
\be\label{eq:multivar-rate}
\begin{split}
\Big| \, \E_{\psi_{N,t}}  f_1 (\cO_{1,t}) \dots f_k (\cO_{k,t}) - &\int dx_1 \dots dx_k \, f_1 (x_1) \dots f_k (x_k) \, \left[ \frac{e^{-\frac{1}{2} \sum_{i,j=1}^k \Sigma^{-1}_{ij} (t) x_i x_j} }{\sqrt{(2\pi)^{k} \det \Sigma (t)}} \, \right]\Big| \\ \leq \; &\frac{C e^{K|t|}}{\sqrt{N}} \prod_{j=1}^k \int d\tau |\widehat{f}_j (\tau)| (1+|\tau|^5 + N^{-1} \tau^8 + N^{-2} \tau^{10}) \end{split} \end{equation}
where $\Sigma^{-1} (t)$ is the inverse of the covariance matrix $\Sigma (t)$.  
\end{thm}

{\it Remarks:} 
\begin{itemize}
\item[i)] The assumptions $\| \ph \|_{H^2} < \infty$ and $\| \partial^\alpha O_j (1-\Delta)^{-|\alpha|/2} \| < \infty$ are needed to control the possible singularity of the interaction potential. If one assumes $V(x)$ to be bounded, the results hold for all $\ph \in H^1 (\bR^3)$ and bounded $O_j$, $j=1,\dots , k$.
\item[ii)] We will show in Section \ref{sect:main} that the products $\langle g_{i,t} , \ph \rangle$ are real, for all $i=1, \dots , k$ and for all $t \in \bR$. Hence 
\[ \text{Re } \Sigma_{ij} (t) = \text{Re } \langle g_{i,t} , g_{j,t} \rangle - \langle g_{i,t} , \ph \rangle \langle \ph, g_{j,t} \rangle \]
and $\text{Im } \Sigma_{ij} (t) = \text{Im } \langle g_{i,t} , g_{j,t} \rangle$ for all $i \leq j$. It is easy to check that the real part $\text{Re } \Sigma (t)$ is non-negative, since
\begin{equation}\label{eq:posit}
\begin{split} \sum_{i,j=1}^k \tau_i \tau_j \text{Re } \Sigma_{ij} (t) &= \text{Re } \left\langle \sum_{i=1}^k \tau_i g_{i,t} , \sum_{j=1}^k \tau_j g_{j,t} \right\rangle - \left\langle \sum_{i=1}^k \tau_i g_{i,t} , \ph \right\rangle \left\langle \ph, \sum_{j=1}^k \tau_j g_{j,t} \right\rangle \\ &= \| g \|^2 - |\langle g , \ph \rangle|^2 \geq 0 
\end{split} \end{equation}
with $g = \sum_j \tau_j g_{j,t}$. The condition that $\text{Re } \Sigma (t)$ is strictly positive is therefore equivalent to the condition that $\ph \not\in \text{span} \{ g_{1,t} , \dots , g_{k,t} \}$. 
\item[iii)] If $\text{Re } \Sigma (t)$ is not strictly positive, then $\Sigma (t)$ does not need to be invertible and (\ref{eq:multivar-rate}) does not hold true. Still, from the proof in Section \ref{sect:main} it follows that 
\begin{equation}\begin{split}  \Big| \, \E_{\psi_{N,t}} f_1 (\cO_{1,t}) \dots f_k (\cO_{k,t}) - &\int d\tau_1, \dots d\tau_k \, \widehat{f}_1 (\tau_1) \dots \widehat{f}_k (\tau_k) \, e^{-\frac{1}{2} \sum_{i,j=1}^k \Sigma_{ij} (t) \tau_i \tau_j} \Big| \\ \leq \; &\frac{C e^{K|t|}}{\sqrt{N}} \prod_{j=1}^k \int d\tau |\widehat{f}_j (\tau)| (1+|\tau|^5 + N^{-1} \tau^8 + N^{-2} \tau^{10}) \end{split} \end{equation}
for all $f_1, \dots , f_k \in L^1 (\bR)$ with $\widehat{f}_j \in L^1 (\bR, (1+ \tau^{10}) d\tau)$.
\item[iv)] Already at time $t= 0$, when particles are independent, the covariance matrix 
\[ \Sigma (0) = \langle O_i \ph , O_j \ph \rangle - \langle \ph, O_i \ph \rangle \langle \ph, O_j \ph \rangle \] has an imaginary part, given by
\[ \text{Im } \Sigma_{ij} (0) = \frac{1}{2i} \langle \ph, [ O_i , O_j ] \ph \rangle \]
for all $i\leq j$. If the operators $O_1, \dots , O_j$ commute, then the imaginary part vanishes, and $\Sigma (0)$ is a real symmetric matrix. In this case, assuming $\Sigma (0)$ to be strictly positive, the integral on the l.h.s. of (\ref{eq:multivar-rate}) is the expectation of the product $\prod_{j=1}^k f_j (x_j)$, where $x_1, \dots , x_k$ are centered Gaussian random variables with covariance matrix $\Sigma (0)$. Hence, for commuting operators $O_1, \dots , O_k$ we recover a ``classical'' multivariate central limit theorem.  
\item[v)] If the operators $O_i$ commute, then the matrix $\Sigma (t)$ remains real symmetric also for times $t \not = 0$. This follows from the properties $U^* (t;0) U(t;0) - V^* (t;0) V(t;0) = 1$ and $U^* (t;0) J V(t;0) J = V^* (t;0) J U (t;0) J$ characterizing the component of a Bogoliubov transformation (see Proposition \ref{prop:bog} below). Here we introduced the antilinear operator $J$ defined by $Jf = \overline{f}$ for all $f \in L^2 (\bR^3)$. To prove that $\Sigma (t)$ is real symmetric we observe that, since $\langle Jf , g \rangle = \langle Jg, f \rangle$ for every $f,g \in L^2 (\bR^3)$, 
\[ \begin{split} \langle g_{i,t} , g_{j,t} \rangle = \;&\langle U(t;0) O_i \ph_t + J V(t;0) O_i\ph_t , U(t;0) O_j \ph_t + J V(t;0) O_j \ph_t \rangle \\ = \; &\langle O_i \ph_t , U^* (t;0) U(t;0)  O_j \ph_t \rangle + \langle V^* (t;0) V(t;0) O_j \ph_t , O_i \ph_t \rangle \\ &+ \langle  O_i \ph_t , U^* (t;0) J V(t;0) O_j \ph_t \rangle + \langle V^* (t;0) J U(t;0) O_j \ph_t , O_i \ph_t \rangle \\ = \; & \langle O_i \ph_t ,O_j \ph_t \rangle + 2\text{Re } \langle O_i \ph_t , V^* (t;0) V(t;0) O_j \ph_t \rangle \\ &+ 2 \text{Re } \langle O_i \ph_t, V^* (t;0) J U(t;0) O_j \ph_t \rangle \end{split} \]
which is clearly real, if $[O_i, O_j ] = 0$. Hence, if the operators $O_1, \dots , O_k$ commute, the integral on the l.h.s. of (\ref{eq:multivar-rate}) is just the expectation $\E_{G_t} f_1 (x_1) \dots f_k (x_k)$ where $G_t$ is a Gaussian vector with real symmetric covariance matrix $\Sigma (t)$.
\item[vi)] In a different setting, a quantum multivariate central limit theorem for the sum of independent and identically distributed random variables has been shown in \cite{JPP}; in this paper the authors 
identify the limiting integral appearing in (\ref{eq:multivar-rate}) as the expectation of $\prod_{j=1}^k f_j (x_j)$ with respect to a quasi-free state.   
\end{itemize}

\medskip

In contrast with the central limit theorem obtained in \cite{BKS}, Theorem \ref{thm:multi-CLT} gives a precise bound on the rate of the convergence towards the Gaussian expectations. For $k = 1$, $\Sigma (t)$ is the scalar
\[ \Sigma (t) = \sigma_t^2 = \| U(t;0) O \ph_t + \overline{V(t;0) O \ph_t} \|^2 - |\langle \ph, U(t;0) O \ph_t + \overline{V(t;0) O \ph_t} \rangle |^2 \] which is always real (and non-negative). Hence, the expectation of $f (\cO_t)$ with respect to the measure induced by $\psi_{N,t}$ converges towards the expectation of $f(x)$, with $x$ a scalar centered Gaussian variable with variance $\sigma_t^2$. We recover in this case the central limit theorem proven in \cite{BKS}. Actually, we obtain more, since we derive also a bound for the convergence rate of probabilities, in the spirit of a Berry-Ess{\'e}en central limit theorem.

\begin{cor}[Berry-Ess{\'e}en type CLT] 
Let $V$ satisfy (\ref{eq:bdV2}). Let $\ph \in H^2 (\bR^3)$ and let $\psi_{N,t} = e^{-i H_N t} \ph^{\otimes N}$ denote the solution of the Schr\"odinger equation (\ref{eq:SE}). Let $O$ be self-adjoint operators on $L^2(\R^3)$, with $\| \partial^{\alpha} O (1-\Delta)^{-|\alpha|/2}\| < \infty$ for all $\alpha \in \bN^3$ with $|\alpha| \leq 2$, and define \[ \cO_t = \frac{1}{\sqrt{N}} \sum_{j=1}^N (O^{(j)} - \langle \ph_t , O \ph_t \rangle) \,. \] 
Then there exists a constant $K > 0$ and, for every $-\infty < \alpha < \beta < \infty$, a constant $C > 0$ such that 
\[ \left| \P_{\psi_{N,t}} (\cO_t  \in [\alpha;\beta]) - \P (G_t \in [\alpha;\beta]) \right| \leq C e^{K|t|} N^{-1/12} \]
where $G_t$ is a centered Gaussian random variable, with variance 
\[ \sigma_t^2 = \| U(t;0) O \ph_t + \overline{V(t;0) O \ph_t} \|^2 - |\langle \ph , U(t;0) O \ph_t + \overline{V(t;0) O \ph_t} \rangle|^2 \, . \]  
\end{cor}

{\it Remark.} The constant $C$ depends on $\alpha, \beta$. {F}rom the proof below, it is clear that it can be bounded by $C \leq c (1+ |\beta-\alpha|)$ for a constant $c > 0$ independent of $\alpha,\beta$. For any $\kappa > 0$, we have
\[ \begin{split} \P_{\psi_{N,t}} (\cO_t  < - N^{\kappa}) \leq \; &\P_{\psi_{N,t}} (N^{-\kappa} |\cO_t| \geq 1) \\  \leq \; &N^{-2\kappa} \, \E_{\psi_{N,t}} \cO_t^2  \\ \leq \; & N^{-2\kappa}\, \tr \, \gamma_{N,t}^{(1)} (O^{(1)} - \langle \ph_t , O \ph_t \rangle)^2 \\ &+ N^{-2\kappa+1}  \, \tr \gamma^{(2)}_{N,t} (O-\langle\ph_t , O \ph_t \rangle) \otimes (O - \langle \ph_t , O \ph_t\rangle)\\ \leq \; &C  \| O \|^2  e^{K|t|} \,  N^{-2\kappa}\end{split} \]
by (\ref{eq:conv-mar}), and similarly for $\P (G < -N^{-\kappa})$. Therefore, we find 
\[ \begin{split}  |\P_{\psi_{N,t}} (\cO_t \leq \beta) - \P (G \leq \beta)|  &\leq C e^{K|t|} N^{-2\kappa} +  |\P_{\psi_{N,t}} (\cO_t \in [-N^{\kappa},\beta]) - \P (G \in [ -N^{\kappa}, \beta])| \\ &\leq C e^{K|t|} N^{-2\kappa} + C e^{K |t|} N^{-1/12+\kappa}\;.  \end{split} \]
Hence, choosing $\kappa = 1/36$, we conclude that, for any $\beta \in \bR$, 
\[ |\P_{\psi_{N,t}} (\cO_t \leq \beta) - \P (G \leq \beta)| \leq C e^{K|t|} N^{-1/18} \,. \]

\begin{proof}
Let $f \in L^1 (\bR)$ with $\widehat{f} \in L^1 (\bR, (1+ \tau^{10}) d\tau)$. It follows from Theorem \ref{thm:multi-CLT} that
\[ \left| \E_{\psi_{N,t}} f (\cO) - \E f (G_t) \right| \leq \frac{C e^{K|t|}}{\sqrt{N}} \int d\tau |\widehat{f} (\tau)| (1+|\tau|^5 + N^{-1} \tau^8 + N^{-2} \tau^{10})  \]
where $G_t$ is a centered Gaussian with variance $\sigma_t^2$. 

Let $\eta \in C_0^\infty (\bR)$ with $\eta (s) \geq 0$ for all $s \in \bR$, $\eta (s) = 0$ for all $|s| > 1$ and $\int \eta (s) ds = 1$. For $\eps > 0$, let $\eta_\eps (s) = \eps^{-1} \eta (s/\eps)$. Let $A=[\alpha, \beta]$. We observe that, for any $\eps > 0$
\[ f_{-,\eps} := \chi_{[\alpha+\eps ; \beta- \eps]} * \eta_\eps \leq \chi_{[\alpha,\beta]}  \leq \chi_{[\alpha-\eps ; \beta+\eps]} *\eta_\eps =: f_{+,\eps} \]
and therefore
\[ \E_{\psi_{N,t}}  f_{-,\eps} (\cO_t) \leq \P_{\psi_{N,t}} (\cO_t \in A) \leq \E_{\psi_{N,t}} f_{+,\eps} (\cO_t)\;. \]
Since \[ \widehat{f_{-,\eps}} (\tau) = \frac{e^{i\tau (\beta-\eps)} - e^{i\tau (\alpha+\eps)}}{i\tau} \, \widehat{\eta} (\eps \tau) \]
we find 
\[ \int d\tau \, (1+ |\tau|^5 + N^{-1} \tau^8 + N^{-2} \tau^{10}) \, |\widehat{f_{-\eps}} (\tau)| \leq C \left( |\beta-\alpha| \eps^{-1} + \eps^{-5} + N^{-1} \eps^{-8} + N^{-2} \eps^{-10} \right) \,. \]
Therefore, we conclude that
\[ \left| \E_{\psi_{N,t}} f_{-,\eps} (\cO_t) - \E f_{-,\eps} (G) \right| \leq C e^{K|t|} (N^{-1/2} \eps^{-5} +  N^{-3/2} \eps^{-8} +  N^{-5/2} \eps^{-10}) \] and, analogously, 
\[ \left| \E_{\psi_{N,t}} \, f_{+,\eps} (\cO_t) - \E f_{+,\eps} (G) \right| \leq C e^{K|t|} (N^{-1/2} \eps^{-5} +  N^{-3/2} \eps^{-8} +  N^{-5/2} \eps^{-10})\;. \]
Hence
\begin{multline*} \E f_{-,\eps} (G) - C e^{K|t|} (N^{-1/2} \eps^{-5} + N^{-3/2} \eps^{-8} + N^{-5/2} \eps^{-10} )
\\  \leq \P_{\psi_{N,t}} (\cO_t \in A) \leq \E f_{+,\eps} (G) + C e^{K|t|} (N^{-1/2} \eps^{-5} +  N^{-3/2} \eps^{-8} +  N^{-5/2} \eps^{-10}) \, . \end{multline*}
Since
\[  \left| \E f_{-,\eps} (G) - \E \chi_A (G) \right| \leq C \eps \quad \text{and } \quad \left| \E f_{+,\eps} (G) - \E \chi_A (G) \right| \leq C \eps \]
we find
\[ \left| \P_{\psi_{N,t}} (\cO_t \in A) - \P (G \in A) \right| \leq C \eps + C e^{K|t|} (N^{-1/2} \eps^{-5} +  N^{-3/2} \eps^{-8} +  N^{-5/2} \eps^{-10})\;. \]
Optimizing over $\eps > 0$ we obtain
\[ \left|  \P_{\psi_{N,t}} (\cO_t \in A) - \P (G \in A) \right| \leq C e^{K|t|} N^{-1/12} \, . \]
\end{proof}

The rest of the paper is organized as follows. In Section \ref{sec:fock} we introduce the formalism of second quantization, we review the main ideas of the coherent states approach developed in \cite{RS} to prove the convergence (\ref{eq:conv-mar}) and we introduce the Bogoliubov transformations $\Theta (t;s)$ describing the limiting evolution of the fluctuations around the 
Hartree dynamics and appearing in the covariance matrix $\Sigma$ of the Gaussian variables in Theorem \ref{thm:multi-CLT}. In Section \ref{sect:est}, we show some key estimates on the growth of the fluctuations around the mean field Hartree dynamics. Using these bounds, we will prove Theorem \ref{thm:multi-CLT} in Section \ref{sect:main}.

\section{Fock space and coherent states approach}\label{sec:fock}
\setcounter{equation}{0}

The bosonic Fock space over $L^2 (\bR^3)$ is defined by
\be\label{eq:fock}
\mathcal{F}=\bigoplus_{n\geq 0}\,L_s^2(\R^{3n})\;.
\ee
%
It is easy to check that $\mathcal{F}$, equipped 
%
%
with the inner product
\bee
\langle \psi\,,\,\phi \rangle = \sum_{n=1}^{\infty} \langle \psi^{(n)}\,,\,\phi^{(n)} \rangle_{L^2} \,,\ \ \ \forall \psi\,,\phi\,\in\ \mathcal{F}\; ,
\eee
is an Hilbert space. The advantage of Fock space, with respect to the $N$-particle space $L^2_s (\bR^{3N})$, is that, on $\cF$, it is possible to consider states where the number of particles is not fixed. A vector $\Psi = \{ \psi^{(0)}, \psi^{(1)}, \psi^{(2)}, \dots \}$ describes a state having $n$ particles with probability $\| \psi^{(n)} \|^2$. 

\medskip

Next, we introduce some important class of operators acting on the Fock space $\cF$. For any operator $O$ on $L^2 (\bR^3)$ (a one-particle operator), we define the second quantization $d\Gamma (O)$ of $O$ by
\bee 
(d\Gamma (O) \psi)^{(n)} = \sum_{j=1}^{n}(1\otimes 1\otimes\dots\otimes O \otimes\dots\otimes 1)\,\psi^{(n)}\, .
\eee
An important example is the number of particle operator $\cN = d\Gamma (1)$ whose action is given by $(\cN \psi)^{(n)} = n \psi^{(n)}$. Notice that, for every bounded operator $O$ on $L^2 (\bR^3)$, we have the operator inequality \[ \pm d\Gamma (O) \leq \| O \| \cN \] and the norm bound $\| d\Gamma (O) \psi \| \leq \| O \| \| \cN \psi \|$.

\medskip

For $f \in L^2 (\bR^3)$, we define the creation operator $a^* (f)$ and its adjoint, the annihilation operator $a(f)$, by 
\bea
&&(a^*(f)\psi)^{(n)}(x_1,\dots,x_n)=\frac{1}{\sqrt n}\sum_{j=1}^n f(x_j)\psi^{(n-1)}(x_1,\dots,x_{j-1},x_{j+1},\dots,x_n)\;,\label{eq:creation} \\
&&(a(f)\psi)^{(n)}(x_1,\dots,x_n)=\sqrt{n+1}\int dx\,\ol{f(x)}\,\psi^{(n+1)}(x,x_1,\dots,x_n)\;. \label{eq:annihilation}
\eea 
Creation and annihilation operators satisfy the canonical commutation relations 
\begin{equation}\label{eq:CCR}
[a(f)\,,\,a^*(g)] = \langle f\,,\,g\rangle\;, \quad \text{ and } \quad [a(f)\,,\,a(g)] = [a^*(f)\,,\,a^*(g)] = 0\;
\end{equation}
for all $f,g \in L^2 (\bR^3)$. We will also use the notation $\phi (f) = a^* (f) + a(f)$. It is also convenient to introduce operator valued distributions $a_x^*, a_x$, which formally create and, respectively, annihilate a particle at point $x$ and are characterized by
\bea
a^*(f)=\int dx\,f(x)\,a_x^*\;, \quad  a(f)=\int dx\,\ol{f(x)}\,a_x\;. \nn
\eea
In terms of these operator valued distributions, on can express the number of particles operator $\cN$ as 
\bee
\Nn=\int dx\,a_x^*\,a_x\;.
\eee
More generally, for a one-particle operator $O$ with integral kernel $O (x,y)$, its second quantization is given by 
\[ d\Gamma (O) = \int dx dy \, O(x,y) a_x^* a_y \, . \]
Observe that creation and annihilation operators are not bounded, but they can be estimated in terms of the square root of the number of particles operator $\Nn$, in the sense that 
\begin{equation}
\begin{split} \label{eq:est-a} 
\| a(f)\psi\| &\leq \|f\|_2 \|\Nn^{1/2}\psi\|\;, \\ 
\| a^*(f)\psi\| &\leq \|f\|_2 \|(\Nn+1)^{1/2}\psi\|\;, \\
\| \phi(f)\psi\| &\leq 2 \|f\|_2 \|(\Nn+1)^{1/2}\psi\|\;,
\end{split} \end{equation}
for all $\psi\in\mathcal F$ and $f\in L^2(\R^3)$. 

In order to define a time-evolution on the Fock space $\cF$, we introduce the Hamilton operator $\cH_N$, by 
\[ \begin{split} (\cH_N\psi)^{(n)} &= \cH_N^{(n)} \psi^{(n)} \quad \text{with } \\
\cH_N^{(n)} &= \sum_{j=1}^n -\Delta_{x_j} + \frac{1}{N} \sum_{i<j}^n V (x_i -x_j)\,. \end{split}  \]
In terms of the operator valued distribution $a_x^*, a_x$, it is simple to check that the Hamiltonian $\cH_N$ can be written as
\begin{equation}\label{eq:ham-F} \cH_N = \int dx \nabla_x a_x^* \nabla_x a_x + \frac{1}{2N} \int dx dy V(x-y) a_x^* a_y^* a_y a_x\,. \end{equation}
We observe that, by definition, the Hamiltonian $\cH_N$ commutes with the number of particles operator (this corresponds to the fact that, for every term in (\ref{eq:ham-F}), the number of creation operators matches the number of annihilation operators). This implies that the time-evolution generated by $\cH_N$ preserves the number of particles in the system. In particular, when restricted to the sector of the Fock space $\cF$ with exactly $N$ particles, the Hamiltonian $\cH_N$ coincides with the $N$-particle Hamiltonian (\ref{eq:ham}). 

\medskip

We will be interested in the time-evolution of so called coherent initial data. For $\f \in L^2 (\bR^3)$, we define the Weyl operator 
\be\label{eq:weyl}
W(\f) = e^{a^*(\f)-a(\f)}\;.
\ee
The coherent state with wave function $\f \in L^2 (\bR^3)$ is defined as $W(\f) \Omega$, where $\Omega = \{ 1, 0 ,0 , \dots \} \in \cF$ is the vacuum. It is easy to check that
\[ W(\f) \Omega = e^{-\| \f \|^2/2} \sum_{n=0}^\infty \frac{(a^* (\f))^n}{n!} \Omega = e^{-\| \f \|^2/2} \left\{ 1 , \f , \frac{\f^{\otimes 2}}{\sqrt{2!}}, \frac{\f^{\otimes 3}}{\sqrt{3!}} , \dots \right\}\,. \]
Coherent states do not have a fixed number of particles; instead they are a linear superposition of states with all possible number of particles. The average number of particles in the coherent state $W(\f) \Omega$ is given by
\[ \langle W(\f) \Omega, \cN W (\f) \Omega \rangle = \| \f \|^2\,. \]
More precisely, the number of particles in a coherent state $W(\f) \Omega$ is a Poisson random variable with average and variance $\| \f \|^2$. This follows from the observation that Weyl operators act as shifts on creation and annihilation operators, in the sense that 
\begin{equation}\label{eq:shifts}
\begin{split}
 W^*(\f)a(f)W(\f) &=a(f)+\langle f\,,\,\f \rangle\;, \\
W^*(\f)a^*(f)W(\f) &=a^*(f)+\langle \f\,,\,f \rangle\;, 
\end{split}
\end{equation}
for all $\f, f\in L^2(\R^3)$. These identities also imply that coherent states are eigenvectors of all annihilation operators, since
\[ a(f) W(\f) \Omega = W(\f) (a(f) + \langle f, \f \rangle) \Omega = \langle f, \f \rangle W(\f) \Omega\,. \]

\medskip

In order to obtain information about the evolution of factorized $N$-particle initial data, we will study the dynamics of coherent states, with average number of particles given by $N$. To this end we fix $\ph \in L^2 (\bR^3)$ with $\| \ph \| = 1$, and we consider the time evolution 
\[ \Psi_{N,t} = e^{-i \cH_N t} W(\sqrt{N} \ph) \Omega\;. \]
Because of the mean field character of the interaction, we may expect that the evolution of an initial coherent state remain approximately coherent, i.e. 
\[ \Psi_{N,t} \simeq W(\sqrt{N} \ph_t) \Omega \]
where $\ph_t$ is the solution of the Hartree equation (\ref{eq:hartree0}). More precisely, we define 
$\xi_{N,t} = W^* (\sqrt{N} \ph_t) \Psi_{N,t}$ so that 
\[ \Psi_{N,t} = W(\sqrt{N} \ph_t) \xi_{N,t}\;. \]
The vector $\xi_{N,t}$ describes the fluctuations around the mean field evolution; $\Psi_{N,t}$ is close to a coherent state, if the number of particles in $\xi_{N,t}$ is small. It is useful to introduce the fluctuation dynamics 
\begin{equation}\label{eq:fluct} \cU_N (t;s) = W^* (\sqrt{N} \ph_t) e^{-i\cH_N (t-s)} W(\sqrt{N} \ph_s) \end{equation}
so that $\xi_{N,t} = \cU_N (t;0) \Omega$. To get convergence towards the Hartree dynamics, and to obtain estimates of the form (\ref{eq:conv-mar}), one need to prove a bound of the form 
\begin{equation}\label{eq:N-bd} \langle \cU_N (t;0) \Omega, \cN \cU_N (t;0) \Omega \rangle  \leq C e^{K |t|} \end{equation}
for the growth of the number of particles with respect to $\cU_N (t;0)$. Such an estimate immediately implies convergence towards the Hartree dynamics for coherent initial data. Projecting coherent states onto the $N$-particle sector of the Fock space, it can also be used to establish the convergence towards Hartree dynamics for approximately factorized $N$-particle initial data; see \cite{RS,CLS}.

\medskip

In order to show a bound of the form (\ref{eq:N-bd}), one observes that the fluctuation dynamics 
$\cU_N (t;s)$ satisfies a Schr\"odinger type equation
\[ i\partial_t \cU_N (t;s) = \cL (t) \cU_N (t;s) \]
with the time-dependent generator
\begin{equation}\label{eq:LN} \begin{split} \cL_N (t) = \; &W^* (\sqrt{N} \ph_t) \cH_N W(\sqrt{N} \ph_t) + \left[ i\partial_t W^* (\sqrt{N} \ph_t) \right] W (\sqrt{N} \ph_t) \\ = \; &\int dx \nabla_x a_x^* \nabla_x a_x + \int dx (V* |\ph_t|^2)(x) a_x^* a_x + \int dx dy V(x-y) \ph_t (x) \overline{\ph}_t (y) a_x^* a_y \\ &+ \frac{1}{2} \int dx dy \, V(x-y) \, \left( \ph_t (x) \ph_t (y) a_x^* a_y^* + \overline{\ph}_t (x) \overline{\ph}_t (y) a_x a_y \right) \\
&+ \frac{1}{\sqrt{N}} \int dx dy V(x-y) a_x^* (\ph_t (y) a_y^* + \overline{\ph}_t a_y) a_x \\ &+ \frac{1}{2N} \int dx dy V(x-y) a_x^* a_y^* a_y a_x \,.\end{split} \end{equation}
In contrast with the original Hamiltonian $\cH_N$, the generator $\cL_N (t)$ of the fluctuation dynamics does not commute with the number of particles operator $\cN$ (more precisely, the terms on the third and fourth line of (\ref{eq:LN}), in which the number of creation operator does not match the number of annihilation operators, do not commute with $\cN$). As a consequence, the number of particles is not conserved by the fluctuation dynamics $\cU_N (t)$. Nevertheless, in \cite{RS} it was possible to prove a bound of the form (\ref{eq:N-bd}) (and similar bounds for higher powers of $\cN$), showing that although the expectation of the number of particles operators grows in time, it remains bounded uniformly in $N$, for every fixed time. 

\medskip

It is worth noticing that this approach does not only prove the convergence (\ref{eq:conv-mar}) towards the limiting evolution with a precise bound on the rate; instead, it also describes the limiting form of the fluctuations around the mean field dynamics. In fact, from the expression (\ref{eq:LN}) for the generator of the fluctuation dynamics, one can expect that, as $N \to \infty$, the evolution of the fluctuations can be approximated by the limiting fluctuation dynamics $\cU_\infty (t;s)$, which solves the Schr\"odinger equation
\begin{equation}\label{eq:Uinfty} i\partial_t \cU_\infty (t;s) = \cL_\infty (t) \cU_\infty (t;s)\end{equation}
with the limiting generator
\begin{equation}\label{eq:Linfty} \begin{split} \cL_\infty (t) = \; &\int dx \nabla_x a_x^* \nabla_x a_x + \int dx (V* |\ph_t|^2)(x) a_x^* a_x + \int dx dy V(x-y) \ph_t (x) \overline{\ph}_t (y) a_x^* a_y \\ &+ \frac{1}{2} \int dx dy \, V(x-y) \, \left( \ph_t (x) \ph_t (y) a_x^* a_y^* + \overline{\ph}_t (x) \overline{\ph}_t (y) a_x a_y \right) \end{split} \end{equation}
obtained from $\cL_N (t)$ by formally taking the limit $N \to \infty$. The following proposition, taken from \cite{GV}, shows the well-posedness of the Schr\"odinger equation (\ref{eq:Uinfty}). 
\begin{prop}[Prop. 4.1 in \cite{GV}]\label{prop:GV}
Let $V \in L^\infty (\bR^3) + L^2 (\bR^3)$, and assume that $t \to \ph_t$ is in $C (\bR, L^2 (\bR^3) \cap L^4 (\bR^3))$ (both conditions hold true under the assumptions of Theorem \ref{thm:multi-CLT}). Then there exists a unique two-parameter group of unitary transformations $\cU_\infty (t;s)$ with $\cU_\infty (s;s) = 1$ for all $s \in \bR$, and such that $\cU_\infty (t;s)$ is strongly differentiable on the domain $\cD (d\Gamma (-\Delta+1))$ with
\begin{equation}\label{eq:Uinfty1} i\partial_t \cU_\infty (t;s) = \cL_\infty (t) \cU_\infty (t;s) \, \end{equation}
where $\cL_\infty (t)$ is the time-dependent generator defined in (\ref{eq:Linfty}).
\end{prop}

Since the limiting generator $\cL_\infty (t)$ is quadratic in creation and annihilation operators, it turns out that the dynamics $\cU_\infty (t;s)$ acts on the Fock space as a Bogoliubov transformation. For $f,g \in L^2 (\bR^3)$, we define $A(f,g) = a(f) + a^* (\overline{g})$. Then, we have the relation
\begin{equation}\label{eq:AA^*} A^* (f,g) = A(\overline{g}, \overline{f}) = A\left( \cJ (f,g) \right), \quad \text{where } \quad \cJ = \left(\begin{array}{ll} 0 & J \\ J & 0 \end{array} \right) \end{equation}
and $J:L^2 (\bR^3) \to L^2 (\bR^3)$ is the antilinear operator defined by $Jf = \overline{f}$. {F}rom (\ref{eq:CCR}), we also find the commutation relations 
\begin{equation}\label{eq:CCR-AA} [ A(f_1,g_1) , A^* (f_2,g_2)] = \left\langle (f_1, g_1), S (f_2, g_2) \right\rangle_{L^2 \oplus L^2} , \quad \text{with } \quad S = \left( \begin{array}{ll} 1 & 0 \\ 0 &-1 \end{array} \right)\,. \end{equation}

A Bogoliubov transformation is a linear map $\nu : L^2 (\bR^3) \oplus L^2 (\bR^3) \to  L^2 (\bR^3) \oplus L^2 (\bR^3)$ which preserves (\ref{eq:AA^*}) and (\ref{eq:CCR-AA}), i.e. $ \nu^*  S \nu =  S$ and $\nu \cJ = \cJ \nu$. It is simple to check that every Bogoliubov transformation has the block form
\[ \nu = \left( \begin{array}{ll} U & JVJ \\  V & JUJ \end{array} \right) \]
where $U,V : L^2 (\bR^3) \to L^2 (\bR^3)$ are linear operators satisfying $U^* U - V^* V = 1$ and $U^* JVJ  - V^* J U J = 0$. 
\begin{prop}[Theorem 2.2 in \cite{BKS}] \label{prop:bog} Let $V^2 (x) \leq D (1-\Delta)$ and $\ph \in H^1 (\bR^3)$. Assume $\cU_\infty$ is defined by (\ref{eq:Uinfty1}). Then, for every $t,s \in \bR$, there exists a Bogoliubov transformation $\Theta (t;s) : L^2 (\bR^3) \oplus L^2 (\bR^3) \to L^2 (\bR^3) \oplus L^2 (\bR^3)$ such that 
\[ \cU_\infty^* (t;s) A(f,g) \cU_\infty (t;s) = A(\Theta (t;s) (f,g)) \, . \] 
Like every Bogoliubov transform, $\Theta (t;s)$ satisfies the relations $\Theta^* (t;s) S \Theta (t;s) = S$ and $\Theta (t;s) \cJ = \cJ \Theta (t;s)$ and it can be decomposed as
\begin{equation}\label{eq:bog-dec} \Theta (t;s) = \left( \begin{array}{ll} U (t;s) & J V(t;s) J \\
V(t;s) & J U(t;s) J \end{array} \right) \end{equation}
for appropriate operators $U (t;s), V(t;s) : L^2 (\bR^3) \to L^2 (\bR^3)$ satisfying $U^* (t;s) U (t;s) - V^* (t;s) V(t;s) = 1$ and $U^* (t;s) J V(t;s) J = V^* (t;s) J U (t;s) J$. Finally, if $\ph_t$ denotes the solution of the Hartree equation (\ref{eq:hartree0}), we have 
\[ \Theta (t;s) (\ph_t, \overline{\ph}_t) = (\ph_s , \overline{\ph}_s) \]
for every $t,s \in \bR$.
\end{prop}
The proof of this proposition can be found in \cite{BKS}. As explained there, the Bogoliubov transformations $\Theta (t;s)$ satisfy the evolution equation
\[ i\partial_t \Theta (t;s) =  \Theta (t;s) \cA (t) \]
with the generator
\[ \cA (t) = \left( \begin{array}{ll} D_t & - JB_tJ \\ B_t & -JD_t J \end{array} \right) \]
with the linear operators 
\[ \begin{split} D_t f &= -\Delta f + (V*|\ph_t|^2) f + (V*\overline{\ph}_t f) \ph_t\,, \\
B_t f &= (V*\overline{\ph}_t f) \overline{\ph}_t\;. \end{split} \]
Observe here that $D_t^* = D_t$ and $B_t^* = J B_t J$ and therefore $\cA (t) = S \cA (t) S$. 

\section{Preliminary estimates}\label{sect:est}
\setcounter{equation}{0}

In this Section we collect some estimates that will be useful in Section \ref{sect:main}. 

First of all, we will need bounds for the growth of moments of the number of particles operator with respect to the fluctuation dynamics $\cU_N (t;s)$. Similar estimates can also be found in \cite{RS,CLS,BKS}, but here we optimize them and simplify their proof.
\begin{prop}\label{lem:UN}
Let $\U_N (t;s)$ be the fluctuation dynamics defined in \eqref{eq:fluct}. For every $j \in \N$, there exist constants $C_j, K_j > 0$ such that
\be\label{eq:lem-UN}
\left\langle \U_N(t;s)  \psi, (\Nn+1)^j \,\U_N(t;s)\,\psi \right\rangle  \leq \, C_j\,e^{K_j|t-s|} \left\langle \psi, (\Nn + 1)^j \left(1+ N^{-1} \Nn \right) \,\psi\right\rangle
\ee
for every $\psi\in\mathcal{F}$, $t\in\R$. Hence 
%
%
\be\label{eq:UN-j}
\|(\Nn+1)^{\frac{j}{2}}\,\U_N (t;s) \,(\Nn^{\frac{j}{2}}+N^{-\frac{1}{2}} \Nn^{\frac{j+1}{2}}+1)^{-1} \| \leq C_j\,e^{K_j\,|t-s|}\;.
\ee
\end{prop}
\begin{proof}
We proceed by induction on $j\in\N$. Without loss of generality, we choose $s=0$ from now on.

\medskip

{\it Step 1: $j=1$.} We compute the time-derivative 
\be
-i\frac{d}{dt} \langle \U_N (t;0) \psi, (\Nn+1)\,\U_N(t;0)\,\psi\rangle 
= \langle \psi\, ,\,\U_N^*(t;0)\,[\mathcal{L}_N(t)\,,\, \Nn]\,\U_N(t;0)\,\psi \rangle\;.  
\ee
Recalling the expression \eqref{eq:LN} for $\mathcal{L}_N(t)$, we find 
\be\label{eq:comm-LN}
\begin{split}
[\mathcal{L}_N(t)\,,\,(\Nn+1)] 
= \; &2 \int dx\,dy\,V(x-y)\,\f_t(x)\,\f_t(y)\,a_x^*\,a_y^* \\ &+ \frac{1}{\sqrt N} \int dx\,dy\,V(x-y)\,\f_t(y)\,a_x\,a_y^*\,a_x - \text{h.c.} \, .
\end{split}
\ee
Using the estimates \eqref{eq:est-a}, we obtain
\[ \begin{split} \Big| \int dx dy V(x-y) \ph_t (x) \ph_t (y) &\langle \psi, a_x^* a_y^* \psi \rangle \Big| \\ &= \left| 
\int dx \ph_t (x) \langle a_x \psi,  a^* (V(x-.)\ph_t) \psi \rangle \right| \\ &\leq \int dx |\ph_t (x)| \, \| a_x \psi \| \, \| a^* (V(x-.) \ph_t) \psi \| \\ &\leq \int dx \| a_x \psi \|^2 + \int dx |\ph_t (x)|^2 \| a^* (V(x-.) \ph_t) \psi \|^2 \\ & \leq (1+ \sup_x \| V(x-.) \ph_t \|_2)  \langle \psi, (\Nn +1) \psi \rangle \end{split} \]
and similarly
\[ \begin{split}
\Big| \frac{1}{\sqrt{N}} \int dx dy V(x-y) \ph_t (y) &\langle \psi, a_x^* a_y^* a_x \psi \rangle \Big| \\ &\leq 
\frac{1}{\sqrt{N}}\int dx dy |V(x-y)| |\ph_t (y)| \| a_x a_y \psi \| \| a_x \psi \| \\ 
&\leq \int dx dy |V(x-y)|^2 |\ph_t (y)|^2 \, \| a_x \psi \|^2 + \frac{1}{N}  \int dx dy \|a_x a_y \psi \|^2 \\
&\leq  \sup_x \| V(x-.) \ph_t \|_2  \langle \psi, \Nn \psi \rangle + \frac{1}{N} \langle \psi , \Nn^2 \psi \rangle\;.\end{split}  \]
Using the fact that $\sup_x \| V(x-.) \ph_t \|_2 \leq C \| \ph_t \|_{H^1}$ is uniformly bounded in time, we conclude that
\be\label{eq:step1}
\begin{split}
\left| \frac{d}{dt} \langle \psi\,,\,\U_N^*(t;0)\, (\Nn+1)\,\U_N(t;0)\,\psi \rangle \right| \leq \; &C \langle \U_N (t;0) \psi, (\Nn+1) \U_N (t;0) \psi \rangle \\ & + \frac{1}{N} \langle \U_N (t;0) \psi, \Nn^2 \U_N (t;0) \psi \rangle\;.
\end{split} \ee
In order to apply Gronwall's Lemma, we need to control the second term. We observe that, using the identities 
\bea
&& W^*(\sqrt{N}\f_t)\,\Nn\,W(\sqrt{N}\f_t) =  \Nn + \sqrt{N}\,\phi(\f_t) + N\;, \label{eq:shift-WN}\\
&&  W^*(\sqrt{N}\f_t)\,\phi(\f_t)\,W(\sqrt{N}\f_t) = \phi(\f_t) + 2\sqrt{N}\; , \label{eq:shift-Wphi}
\eea
we have
\be\label{eq:N2}
\begin{split}
\frac{1}{N} \big\langle \U_N(t;0) \,&\psi, \Nn^2 \U_N(t;0) \,\psi \big\rangle \\ =\; &\frac{1}{N} \big\langle \Nn\,\U_N(t;0)\,\psi\,,\,W^*(\sqrt{N}\f_t)\,(\Nn-\sqrt{N}\phi(\f_t) + N)\,e^{-i\mathcal{H}_N t} \,W(\sqrt{N}\f)\,\psi \big\rangle \\ = \; &\frac{1}{N} \langle \Nn\,\U_N(t;0)\,\psi\,,\,W^*(\sqrt{N}\f_t)\,e^{-i\mathcal{H}_N t} \Nn W(\sqrt{N} \ph) \psi \rangle \\ &- \frac{1}{\sqrt N} \langle \Nn\,\U_N(t;0)\,\psi\,,\,W^*(\sqrt{N}\f_t)\,\phi(\f_t)\, e^{-i \cH_Nt} W(\sqrt{N} \ph) \psi \rangle \\  &+ \big\langle \U_N(t;0)\,\psi,  \Nn \U_N(t;0)\,\psi\big\rangle\;.
\end{split}
\ee
We rewrite the first term on the r.h.s. as
\[ \begin{split} 
\frac{1}{N} \langle \Nn\,\U_N(t;0)\,&\psi\,,\,W^*(\sqrt{N}\f_t)\,e^{-i\mathcal{H}_N t} \Nn W(\sqrt{N} \ph) \psi \rangle \\ = \; & \frac{1}{N} \langle \Nn\,\U_N(t;0)\,\psi\,,\U_N (t;0) ( \Nn + \sqrt{N} \phi (\ph) + N) \psi \rangle  \\ = \; & \langle \U_N(t;0)\,\psi , \Nn \U_N(t;0)\,\psi \rangle + \frac{1}{N} \langle \Nn\,\U_N(t;0)\,\psi\,,\U_N (t;0) \Nn \psi \rangle \\
&+\frac{1}{\sqrt{N}} \langle \Nn\,\U_N(t;0)\,\psi\,,\U_N (t;0) \phi (\ph) \psi \rangle\;.
\end{split} \]
This implies that
\be\label{eq:bd1} \begin{split} \Big| \frac{1}{N} \langle \Nn\,\U_N(t;0)\,\psi\,,\,W^*(\sqrt{N} &\f_t)\,e^{-i\mathcal{H}_N t} \Nn W(\sqrt{N} \ph) \psi \rangle \Big| \\
 \leq \; & \langle \U_N(t;0)\,\psi , \Nn \U_N(t;0)\,\psi \rangle + \frac{1}{N} \| \Nn  \U_N(t;0)\,\psi \| \| \Nn \psi \| \\ &+ \frac{1}{\sqrt{N}} \| \Nn  \U_N(t;0)\,\psi \|  \| \phi (\ph) \psi \| \\
 \leq \; & \langle \U_N(t;0)\,\psi , \Nn \U_N(t;0)\,\psi \rangle + \frac{1}{4N} \|  \Nn  \U_N(t;0)\,\psi \|^2 \\ &+ C \langle \psi , \left( \Nn + N^{-1} \Nn^2 +1 \right) \psi \rangle\;.
 \end{split} \ee
As for the second term on the r.h.s. of (\ref{eq:N2}), we find
\[ \begin{split}  \frac{1}{\sqrt N} \langle \Nn\,&\U_N(t;0)\,\psi\,,\,W^*(\sqrt{N}\f_t)\,\phi(\f_t)\, e^{-i \cH_Nt} W(\sqrt{N} \ph) \psi \rangle \\ = \; & \frac{1}{\sqrt N} \Big\langle \Nn\,\U_N(t;0)\,\psi\,, \,\left( \phi(\f_t) + 2\sqrt{N} \right) \U_N (t;0) \psi \Big\rangle \\
= &\; 2\langle \U_N(t;0)\,\psi , \Nn \U_N(t;0)\,\psi \rangle + \frac{1}{\sqrt{N}}   \langle \Nn\,\U_N(t;0)\,\psi\,, \phi(\f_t)  \U_N (t;0) \psi \rangle \end{split} \]
which leads to
\be\label{eq:bd2} \begin{split} 
\Big|  \frac{1}{\sqrt N} \langle \Nn\, \U_N(t;0)\,\psi\,,\,W^*(&\sqrt{N}\f_t)\,\phi(\f_t)\, e^{-i \cH_Nt} W(\sqrt{N} \ph) \psi \rangle \Big| \\ \leq \; &  3\langle \U_N(t;0)\,\psi , (\Nn+1)  \U_N(t;0)\,\psi \rangle + \frac{1}{4N} \| \Nn \U_N(t;0)\,\psi \|^2\;. \end{split} \ee
Combining (\ref{eq:bd1}) and (\ref{eq:bd2}), we conclude from (\ref{eq:N2}) that
\[ \begin{split} \frac{1}{N} \big\langle \U_N(t;0) \,&\psi, \Nn^2 \U_N(t;0) \,\psi \big\rangle \\ \leq \; &\frac{1}{2N} \big\langle \U_N(t;0) \,\psi, \Nn^2 \U_N(t;0) \,\psi \big\rangle +  4\langle \U_N(t;0)\,\psi , (\Nn+1)  \U_N(t;0)\,\psi \rangle \\ &+ C \langle \psi, \left( \Nn + N^{-1} \Nn^2 \right) \psi\rangle\;. \end{split} \]
Subtracting the first term back on the l.h.s. gives 
\[ \begin{split} \frac{1}{N} \big\langle \U_N(t;0) \,&\psi, \Nn^2 \U_N(t;0) \,\psi \big\rangle \\ \leq \; & 8 \langle \U_N(t;0)\,\psi , (\Nn+1)  \U_N(t;0)\,\psi \rangle + C \langle \psi, \left( \Nn + N^{-1} \Nn^2 \right) \psi\rangle\;. \end{split} \]
Inserting the last estimate in (\ref{eq:step1}), we find
\[\begin{split}
\left| \frac{d}{dt} \langle \psi\,,\,\U_N^*(t;0)\, (\Nn+1)\,\U_N(t;0)\,\psi \rangle \right| \leq \; &C \langle \U_N (t;0) \psi, (\Nn+1) \U_N (t;0) \psi \rangle \\ & + C \langle \psi, (\Nn + N^{-1} \Nn^2 ) \psi \rangle\;. 
\end{split} \]
Gronwall's Lemma implies that
\be
\langle \psi\,,\,\U_N^*(t;0)\,(\Nn+1)\;\U_N(t;0)\,\psi \rangle \leq C e^{K t} \langle \psi\,,(\Nn+ N^{-1} \Nn^2+1)\,\psi \rangle\, \;,
\ee
for appropriate constants $C,K > 0$.

\medskip

{\it Step 2:} we assume 
\be\label{eq:indu}
\left\langle \U_N(t;0)  \psi, (\Nn+1)^{i} \,\U_N(t;0)\,\psi \right\rangle  \leq \, C_{i} \,e^{K_{i} t} \left\langle \psi, \Nn^{i} \left(1+ N^{-1} \Nn \right) \,\psi\right\rangle
\ee
for all $i \leq (j-1)$ and we prove it for $i=j$ (this is exactly (\ref{eq:lem-UN})).

\medskip

{F}rom 
\be
[\mathcal{L}_N\,,\,\Nn^j]=\sum_{i=1}^j \Nn^{i-1}\,[\mathcal{L}_N\,,\,\Nn]\,\Nn^{j-i}
\ee
we find 
\be\label{eq:est-step2}
\begin{split} 
-i \frac{d}{dt} \big\langle \U_N(t;0)\,\psi, &(\Nn+1)^j  \U_N(t;0)\,\psi \big\rangle  \\ &= 
\sum_{i=1}^j \langle \psi\,,\,\U_N^*(t;0)\,(\Nn + 1)^j \,[\mathcal{L}_N\,,\,\Nn]\,(\Nn+1)^{j-i-1}\,\U_N(t;0)\,\psi \rangle\;.
\end{split}
\ee
{F}rom \eqref{eq:comm-LN}, arguing as in Step 1 and using the intertwining relations 
\bea
&& \Nn a(f) = a(f)\,(\Nn-1), \quad \text{and }  \Nn a^*(f) = a^*(f)\,(\Nn+1)\,,
\eea 
we find
\be\label{eq:est-step2-1}
\begin{split}
\Big| \frac{d}{dt} \big\langle &\U_N(t;0)\,\psi, (\Nn+1)^j  \U_N(t;0)\,\psi \big\rangle \Big| \\ &\leq C \langle \U_N(t;0)\,\psi,  (\Nn + 1)^j \U_N(t;0)\,\psi \rangle + \frac{C}{N} \langle \U_N (t;0) \psi , (\Nn+1)^{j+1} \U_N (t;0) \psi \rangle\;. \end{split} \end{equation}
In order to apply Gronwall's Lemma, we have to estimate the second term on the r.h.s. of \eqref{eq:est-step2-1}. We claim that for all $i \leq j$, there exist constants $C,K > 0$ such that
\begin{equation}\label{eq:claim2}
\begin{split}  \frac{1}{N} \langle \U_N (t;0) \psi, (\cN+1)^{i+1} \U_N (t;0) \psi \rangle \leq \; &C e^{K t} \langle \psi, (\cN+1)^i (1+ N^{-1} \cN) \psi \rangle \\ &+ C \langle \U_N (t;0) \psi, (\cN+1)^{i} \U_N (t;0) \psi \rangle\;. \end{split} 
\end{equation}
Inserting (\ref{eq:claim2}) into the r.h.s. of (\ref{eq:est-step2-1}) with $i=j$ and applying Gronwall inequality, we obtain (\ref{eq:lem-UN}). 

In order to prove (\ref{eq:claim2}), we proceed again by induction. For $i=1$, (\ref{eq:claim2}) was proven in Step 1. Similarly, one can show (\ref{eq:claim2}) for $i=0$ (the proof is simpler in this case, one just need to observe that $W^*(\sqrt{N} \ph_t) (\cN+1) W(\sqrt{N} \ph_t) = (\cN+\sqrt{N} \phi (\ph_t) + N + 1) \leq 2 (\cN + N + 1)$ which then commutes with the evolution $\exp (-i\cH_N t)$). We assume hence that (\ref{eq:claim2}) holds for $i=k-1$ and we show it for $i=k \in \bN$, for an arbitrary $2 \leq k \leq j$. 

Using \eqref{eq:shift-WN} and \eqref{eq:shift-Wphi}, similarly
as in Step 1, we obtain
\be\label{eq:est-step2-2}
\begin{split} 
\frac{1}{N}\langle \psi\,,\, &\U_N^*(t;0)\,(\Nn+1)^{k+1}\,\U_N(t;0)\,\psi \rangle \\ =\; &\frac{1}{N}\langle (\Nn+1)^k \U_N (t;0) \psi\,,\, \,(\Nn+1)\,W^*(\sqrt{N} \ph_t) e^{-i\cH_N t} W (\sqrt{N} \ph) \psi \rangle\\
=\; &\frac{1}{N}\langle (\Nn+1)^k \U_N (t;0) \psi\,,\, \,W^*(\sqrt{N} \ph_t)  (\Nn - \sqrt{N} \phi (\ph_t) + N+1) \,e^{-i\cH_N t} W (\sqrt{N} \ph) \psi \rangle\\
=\; &\frac{1}{N}\langle (\Nn+1)^k \U_N (t;0) \psi\,,\, \,W^*(\sqrt{N} \ph_t)   \,e^{-i\cH_N t} (\Nn +1) W (\sqrt{N} \ph) \psi \rangle\\
\; & - \frac{1}{\sqrt{N}}\langle (\Nn+1)^k \U_N (t;0) \psi\,,\,  \phi (\ph_t)  \, \U_N (t;0)  \psi \rangle\\
\; &- \langle  \U_N (t;0) \psi, \, (\Nn+1)^k \U_N (t;0) \psi \rangle
\end{split} \ee
where, in the last step, we commuted the operator $\phi(\varphi_t)$ through the Weyl operator $W^\ast(\sqrt{N}\varphi_t)$. In the first term on the r.h.s. of the last equation, we move the single number of particles operator $(\Nn +1)$ to the right of the Weyl operator $W(\sqrt{N} \ph)$. We find
\begin{equation}\label{eq:step2-3} \begin{split} 
\frac{1}{N}\langle \psi\,,\, \U_N^*(t;0)\,(\Nn+1)^{k+1}\,&\U_N(t;0)\,\psi \rangle \\ =\; &\frac{1}{N}\langle (\Nn+1)^k \U_N (t;0) \psi\,,\, \,\U_N (t;0) (\Nn+1) \psi \rangle \\& + \frac{1}{\sqrt{N}} \langle (\Nn+1)^k \U_N (t;0) \psi\,,\, \,\U_N (t;0) \phi (\ph)  \psi \rangle \\
& - \frac{1}{\sqrt{N}}\langle (\Nn+1)^k \U_N (t;0) \psi\,,\,  \phi (\ph_t)  \, \U_N (t;0)  \psi \rangle\;.
\end{split} \end{equation}
The third term on the r.h.s. of the last equation can be estimated by
\[ \begin{split} 
&\left| \frac{1}{\sqrt{N}}  \langle (\Nn+1)^k \U_N (t;0) \psi\,,\,  \phi (\ph_t)  \, \U_N (t;0)  \psi \rangle \right| 
\\ &\hspace{1cm} \leq \alpha \langle \U_N (t;0) \psi, (\Nn+1)^k \U_N (t;0) \psi\rangle + \frac{1}{\alpha N} \langle \U_N (t;0) \psi, \phi (\ph_t) (\cN+1)^k \phi (\ph_t) \U_N (t;0) \psi \rangle \\
&\hspace{1cm} \leq \alpha \langle \U_N (t;0) \psi, (\Nn+1)^k \U_N (t;0) \psi\rangle + \frac{C}{\alpha N} \langle \U_N (t;0) \psi, (\cN+1)^{k+1} \U_N (t;0) \psi \rangle \end{split} \]
where $\alpha > 0$ is arbitrary and where we used the fact that, for every $k \in \bN$, there exists a constant $C$ such that 
\begin{equation}\label{eq:phiNphi} \phi (\ph_t) (\cN+1)^k \phi (\ph_t) \leq C (\cN+1)^{k+1}\;. \end{equation}
Choosing $\alpha > 0$ sufficiently large, we find
\begin{equation}\label{eq:est-step2-3-1}\begin{split} &\left|\frac{1}{\sqrt{N}} \langle (\Nn+1)^k \U_N (t;0) \psi\,,\,  \phi (\ph_t)  \, \U_N (t;0)  \psi \rangle \right| 
\\ &\hspace{1cm} \leq C \langle \U_N (t;0) \psi, (\Nn+1)^k \U_N (t;0) \psi\rangle + \frac{1}{4 N} \langle \U_N (t;0) \psi, (\cN+1)^{k+1} \U_N (t;0) \psi \rangle\;. \end{split} \end{equation}
The second term on the r.h.s. of (\ref{eq:step2-3}) can be bounded by
\[ \begin{split} &\left| \frac{1}{\sqrt{N}} \langle (\Nn+1)^k \U_N (t;0) \psi\,,\, \,\U_N (t;0) \phi (\ph)  \psi \rangle \right| \\ & \leq \frac{1}{4N} \langle \U_N (t;0) \psi, (\cN +1)^{k+1} \U_N (t;0) \psi \rangle + \langle \U_N (t;0) \phi (\ph)  \psi , (\cN+1)^{k-1} \U_N (t;0) \phi (\ph) \psi \rangle\;.
\end{split} \]
{F}rom the induction assumption (\ref{eq:indu}) we obtain, using again (\ref{eq:phiNphi}),
\begin{equation}\label{eq:est-step2-3-2} \begin{split} &\left| \frac{1}{\sqrt{N}} \langle (\Nn+1)^k \U_N (t;0) \psi\,,\, \,\U_N (t;0) \phi (\ph)  \psi \rangle \right| \\ &\hspace{1cm} \leq \frac{1}{4N} \langle \U_N (t;0) \psi, (\cN +1)^{k+1} \U_N (t;0) \psi \rangle \\ &\hspace{1.4cm} + C e^{K t}  \langle \psi , \phi (\ph) (\cN+1)^{k-1} (1+ N^{-1} \cN) \phi (\ph) \psi \rangle  \\ &\hspace{1cm} \leq \frac{1}{4N} \langle \U_N (t;0) \psi, (\cN +1)^{k+1} \U_N (t;0) \psi \rangle + C e^{K t}  \langle \psi , (\cN+1)^{k} (1+ N^{-1} \cN) \psi \rangle\;. \end{split} \end{equation}
Finally, we control the first term on the r.h.s. of (\ref{eq:step2-3}). To this end, we need to commute one more factor $(\cN+1)$ across the fluctuation evolution $\cU_N (t;0)$. We write,
similarly to~\eqref{eq:step2-3},
\begin{align}\begin{split}\label{eq:k+1term}
\frac{1}{N}\langle (\Nn+1)^k \U_N (t;0) \psi\,,\, &\,\U_N (t;0) (\Nn+1) \psi \rangle \\
 =\; &\frac{1}{N}\langle (\Nn+1)^{k-1} \U_N (t;0)(\Nn+1) \psi\,,\, \,\U_N (t;0) (\Nn+1) \psi \rangle \\
 & + \frac{1}{\sqrt{N}} \langle (\Nn+1)^{k-1} \U_N (t;0)\phi(\ph) \psi\,,\, \,\U_N (t;0)(\Nn+1) \psi \rangle \\
& - \frac{1}{\sqrt{N}}\langle (\Nn+1)^{k-1} \phi (\ph_t) \U_N (t;0) \psi\,,\,   \, \U_N (t;0) (\Nn+1) \psi \rangle\;.\\ 
\end{split}
\end{align}
Using the induction hypothesis \eqref{eq:claim2} with $i = k-2$, the first term on the r.h.s. of the last equation can be estimated by
\begin{align}\label{eq:term1}
 \begin{split}
  \frac{1}{N} \langle \U_N (t;0) (\Nn +1) \psi &, (\Nn + 1)^{k-1} \,  \U_N (t;0) (\Nn+1) \psi \rangle\\
  \leq&
 C e^{K t} \langle  \psi , (\Nn + 1)^{k} \,   (1+N^{-1} \Nn) \psi \rangle\\
  &+C\langle \U_N (t;0) (\Nn +1) \psi , (\Nn + 1)^{k-2} \,  \U_N (t;0) (\Nn+1) \psi \rangle\\
  \leq&C e^{K t} \langle  \psi , (\Nn + 1)^{k} \,   (1+ N^{-1} \Nn) \psi \rangle
 \end{split}
\end{align}
where in the last inequality we also used the assumption (\ref{eq:indu}), with $i=k-2$. The second term on the r.h.s. of (\ref{eq:k+1term}) can be bounded using (\ref{eq:indu}) by 
\begin{align*}
 \begin{split}
 \Big|  \frac{1}{\sqrt{N}} \langle (\Nn+1)^{k-1} \U_N &(t;0)  \phi (\ph) \psi ,  \U_N (t;0) (\Nn+1) \psi \rangle \Big| \\
  \leq&\langle \mathcal{U}_N(t;0)\phi(\varphi)\psi,(\Nn+1)^{k-1}\mathcal{U}_N(t;0)\phi(\varphi)\psi\rangle\\
  &+ \frac{1}{N} \langle \U_N (t;0) (\Nn +1) \psi , (\Nn + 1)^{k-1} \,  \U_N (t;0) (\Nn+1) \psi \rangle\\
  \leq&Ce^{K t}\langle\psi,(\Nn+1)^{k}(1+\Nn/N)\psi\rangle\\
  &+ \frac{1}{N} \langle \U_N (t;0) (\Nn +1) \psi , (\Nn + 1)^{k-1} \,  \U_N (t;0) (\Nn+1) \psi \rangle.
 \end{split}
\end{align*}
{F}rom (\ref{eq:term1}), we find
\begin{equation}\label{eq:term2} \begin{split} 
 \left| \frac{1}{\sqrt{N}} \langle (\Nn+1)^{k-1} \U_N (t;0)  \phi (\ph) \psi ,  \U_N (t;0) (\Nn+1) \psi \rangle \right|  \leq  Ce^{K t}\langle\psi,(\Nn+1)^{k}(1+\Nn/N)\psi\rangle\;.
   \end{split} \end{equation}
Similarly, using (\ref{eq:term1}) and the bound (\ref{eq:phiNphi}) the third term on the r.h.s. of \eqref{eq:k+1term} is bounded by
\begin{align}\label{eq:term3}
 \begin{split}
\Big|  \frac{1}{\sqrt{N}} \langle &(\Nn+1)^{k-1} \phi (\ph_t) \U_N (t;0) \psi ,  \U_N (t;0) (\Nn+1) \psi \rangle \Big| \\
  \leq&
  \frac{1}{N} \langle \U_N (t;0) (\Nn +1) \psi , (\Nn + 1)^{k-1} \,  \U_N (t;0) (\Nn+1) \psi \rangle\\
  &+\langle \U_N (t;0) \psi ,  \phi(\varphi_t) (\Nn+1)^{k-1} \phi (\ph_t)\U_N (t;0) \psi \rangle\\
  \leq& Ce^{K t}\langle\psi,(\Nn+1)^{k}(1+\Nn/N)\psi\rangle +C\langle \U_N (t;0) \psi , (\Nn+1)^{k}\U_N (t;0) \psi \rangle\;.
 \end{split}
\end{align}
Combining (\ref{eq:term1}), \eqref{eq:term2} and \eqref{eq:term3}, we obtain from (\ref{eq:k+1term}) that
\[ \begin{split} \left| \frac{1}{N} \langle (\Nn+1)^k \U_N (t;0) \psi\,,\, \U_N (t;0) (\Nn+1) \psi \rangle \right|\leq\; & C
e^{K t}\langle\psi,(\Nn+1)^{k}(1+\Nn/N)\psi\rangle \\ & +C\langle \U_N (t;0) \psi , (\Nn+1)^{k}\U_N (t;0) \psi \rangle\;. \end{split} \]
This, together with (\ref{eq:est-step2-3-1}) and (\ref{eq:est-step2-3-2}), gives the following bound for (\ref{eq:step2-3}):
\[ \begin{split}
\frac{1}{N}\langle \psi\,,\, &\U_N^*(t;0)\,(\Nn+1)^{k+1}\,\U_N(t;0)\,\psi \rangle \\ \leq \; &\frac{1}{2N} \langle 
 \psi\,,\, \U_N^*(t;0)\,(\Nn+1)^{k+1}\, \U_N(t;0)\,\psi\rangle \\ &+ C \langle \cU_N (t;0) \psi, (\cN+1)^k \cU_N (t;0) \psi \rangle + C e^{K t} \langle \psi, (\cN+1)^k (1+N^{-1} \cN) \psi \rangle\;. \end{split} \]
Subtracting the first term back in the l.h.s, we find 
 \[ \begin{split} 
\frac{1}{N}\langle \psi\,,\, &\U_N^*(t;0)\,(\Nn+1)^{k+1}\,\U_N(t;0)\,\psi \rangle \\ &\leq C \langle \cU_N (t;0) \psi, (\cN+1)^k \, \cU_N (t;0) \psi \rangle + C e^{K t} \langle \psi, (\cN+1)^k (1+N^{-1} \cN) \psi \rangle \end{split} \]
which proves (\ref{eq:claim2}), for $i=k$. 
\end{proof}

We will also need similar bounds for the growth of moments of the number of particles operator, of the kinetic energy operator and of the square of the kinetic energy operator with respect to the limiting fluctuation dynamics $\cU_\infty (t;s)$. The proof of the following lemma can be found in \cite{CLS}[Prop. 4.1]. 
\begin{lem}\label{lem:U2}
Let $\cU_\infty (t;s)$ be the limiting fluctuation dynamics, defined in Proposition \ref{prop:GV}. For every $j \in \N$, there exist constants $C, K > 0$ (depending on the constant $D$ appearing in (\ref{eq:bdV2}), on $\| \ph \|_{H^1}$ and on $j$) such that
\be
\langle \psi, \cU^*_\infty (t;s) \, (\Nn+1)^{j} \,\U_\infty(t;s)\,\psi\rangle \, \leq \, C \,e^{K |t-s|} \langle\psi, (\Nn+1)^{j}\,\psi\rangle
\ee
for all $\psi\in\mathcal{F}$, $t,s\in\R$. 
%
Let moreover \[ \cK = d\Gamma (-\Delta) = \int dx \nabla_x a_x^* \nabla_x a_x \]
denote the kinetic energy operator. Then there exist constants $C,K > 0$ (depending on $D$ and $\| \ph \|_{H^1}$) and $C',K' > 0$ (depending on $D$ and $\| \ph \|_{H^2}$) such that
\[ \langle \psi, \cU_\infty^* (t;s) \, \cK \, \cU_\infty (t;s) \psi \rangle \leq C e^{K |t-s|} \langle \psi, (\cK+\cN +1) \psi \rangle \] 
and 
\[ \langle \psi, \cU_\infty^* (t;s) \, \cK^2 \, \cU_\infty (t;s) \psi \rangle \leq C' e^{K' |t-s|} \langle \psi, (\cK^2 + \cN^2+1) \psi \rangle \] 
for all $\psi\in\mathcal{F}$, $t,s\in\R$. 
\end{lem}

Next, we will need to compare the fluctuation dynamics $\cU_N (t;s)$ with its formal limit $\cU_\infty (t;s)$. To this end, we will make use of the following proposition.

\begin{prop}\label{prop:U-Uin}
Let $\cU_N (t;s)$ be the fluctuation dynamics defined in \eqref{eq:fluct} and let $\U_\infty (t;s)$ be defined as in Proposition \ref{prop:GV}. Then, for every $j \in \bN$ there exists constants $C_j, K_j > 0$ (depending on the constant $D$ appearing in (\ref{eq:bdV2}), on $\| \ph \|_{H^1}$ and on $j$) such that
\[ \begin{split} \big\| (\cN+1)^{j/2} &\left(\cU_N (t;s) - \cU_\infty (t;s) \right) \psi \big\| \\ &\leq \frac{C_je^{K_j|t-s|}}{\sqrt{N}} \left( \| (\cN+1)^{(j+3)/2} \psi \| + \| \cK \psi \| + \frac{1}{N} \| (\cN+1)^{j+3} (1+ N^{-1} \cN) \psi \| \right)\;. \end{split} \]
\end{prop}
\begin{proof}
{F}rom (\ref{eq:LN}) and (\ref{eq:Linfty}) we find  
 \bee
 \begin{split}
\left(
  \mathcal{U}_N(t;s)-\mathcal{U}_\infty(t;s)\right)\psi &=-i \int_{s}^t \mathrm{d}r \;\mathcal{U}_N(t;r)
  \left(\mathcal{L}_N (r)-\mathcal{L}_\infty(r)\right)\mathcal{U}_\infty(r;s)\psi \\
  &=   -i \int_{s}^t \mathrm{d}r \; \mathcal{U}_N(t;r)\left(\mathcal{L}_3(r)+\mathcal{L}_4\right)\mathcal{U}_\infty(r;s)\psi
 \end{split}
 \eee
with
\[ \begin{split} \cL_3 (t) &= \frac{1}{\sqrt{N}} \int dx dy\; V(x-y) \, a_x^* \left(\ph_t (y)  a_y^* +\overline{\ph}_t (y)  a_y \right) a_x \quad \text{and } \\
\cL_4 &= \frac{1}{2N} \int dx dy\; V(x-y) a_x^* a_y^* a_y a_x \, . \end{split}  \]
{F}rom Proposition \ref{lem:UN} we find 
 \be\label{eq:difference1}
 \begin{split}
& \left\Vert (\mathcal{N}+1)^{j/2}\left(\mathcal{U}_N(t;s)-\mathcal{U}_\infty(t;s)\right)\psi\right\Vert \\
 &\hspace{2.5cm}\leq \int_s^t\mathrm{d}r \;\left\Vert 
 (\mathcal{N}+1)^{j/2} \mathcal{U}_N(t;r)\left(\mathcal{L}_3(r)+\mathcal{L}_4\right)\mathcal{U}_\infty(r;s)\psi\right\Vert\\
 &\hspace{2.5cm} \leq C\int_s^t \mathrm{d}r\; e^{K|t-r|}\left(\left\Vert (\mathcal{N}+1)^{j/2}(1+\mathcal{N}/N)^{1/2}\;\mathcal{L}_3(r)\;\mathcal{U}_\infty(r;s)\psi\right\Vert\right. \\
&\hspace{6.5cm} \left. +\left\Vert (\mathcal{N}+1)^{j/2}(1+\mathcal{N}/N)^{1/2}\;\mathcal{L}_4\;\mathcal{U}_\infty(r;s)\psi\right\Vert\right).
 \end{split}
 \ee
Using the estimate
\[ \cL_3 (r) (\cN+1)^j \cL_3 (r) \leq \frac{C}{N} \, (\cN+1)^{j+3} \]
proven in \cite{CLS}[Lemma 6.3], the term containing $\cL_3 (r)$ on the r.h.s. of (\ref{eq:difference1}) can be bounded by 
\be\label{eq:difference2}
 \begin{split}
 &\left\Vert (\mathcal{N}+1)^{j/2}(1+\mathcal{N}/N)^{1/2}\mathcal{L}_3(r)\mathcal{U}_\infty(r;s)\psi\right\Vert\\
 &\hspace{2.5cm}\leq \frac{C}{\sqrt{N}}\left\Vert (\mathcal{N}+1)^{(j+3)/2}(1+\mathcal{N}/N)^{1/2}\mathcal{U}_\infty(r;s)\psi\right\Vert\\
 &\hspace{2.5cm}\leq \frac{C}{\sqrt{N}}e^{K|r-s|}\left\Vert (\mathcal{N}+1)^{(j+3)/2}(1+\mathcal{N}/N)^{1/2}\psi\right\Vert
 \end{split}
 \ee
where we also applied Lemma \ref{lem:U2} to control the growth of powers of $\cN$ w.r.t. 
$\cU_\infty (r;s)$. To bound the term containing $\cL_4$ on the r.h.s. of (\ref{eq:difference1}), on the other hand, we use 
\begin{equation}\label{eq:bd-L4} (\cN+1)^{j/2} \cL_4^2 (\cN+1)^{j/2} \leq \frac{C}{N^2} (\cN+1)^{j+3} (\cN+\cK) \end{equation}
This estimate can be shown considering the restriction of the operator on the l.h.s. on the $n$-particle sector $\cF_n$. {F}rom (\ref{eq:bdV2}), we conclude that
\[ \begin{split} 
 (\cN+1)^{j/2} \cL_4^2 (\cN+1)^{j/2} |_{\cF_n} &= \frac{(n+1)^{j+4}}{N^2} \left[ \frac{1}{(n+1)^2} \sum_{i<j}^n V(x_i -x_j) \right]^2 \\ & \leq \frac{(n+1)^{j+2}}{N^2} \sum_{i<j}^n V^2 (x_i -x_j) \\ 
 &\leq  C \frac{(n+1)^{j+3}}{N^2} \sum_{j=1}^n (1-\Delta_{x_j})  
\end{split} \] 
which is exactly the restriction of the r.h.s. of (\ref{eq:bd-L4}) on $\cF_n$. With (\ref{eq:bd-L4}), we can bound the term containing $\cL_4$ on the r.h.s. of (\ref{eq:difference1}) by
\be\label{eq:difference3}
 \begin{split}
 \left\Vert (\mathcal{N}+\right. & \left.1)^{j/2}(1 +\mathcal{N}/N)^{1/2}\; \mathcal{L}_4 \;\mathcal{U}_\infty(r;s)\psi\right\Vert\\
 \leq& \frac{C}{N}\left\Vert (\mathcal{N}+1)^{(j+4)/2}(1+\mathcal{N}/N)^{1/2}\mathcal{U}_\infty(r;s)\psi\right\Vert\\
 & +\frac{C}{N}\left\Vert (\mathcal{N}+1)^{(j+3)/2}(1+\mathcal{N}/N)^{1/2}\mathcal{K}^{1/2}\mathcal{U}_\infty(r;s)\psi\right\Vert  \\
 \leq & \frac{C}{\sqrt{N}}\left\Vert (\mathcal{N}+1)^{(j+3)/2}\mathcal{U}_\infty(r;s)\psi\right\Vert
 +\frac{C}{N^{3/2}}\left\Vert (\mathcal{N}+1)^{(j+5)/2}(1+\mathcal{N}/N)\mathcal{U}_\infty(r;s)\psi\right\Vert\\
   & +\frac{C}{\sqrt{N}}\left\Vert \mathcal{K}\mathcal{U}_\infty(r;s)\psi\right\Vert
 +\frac{C}{N^{3/2}}\left\Vert (\mathcal{N}+1)^{(j+3)}(1+\mathcal{N}/N)\mathcal{U}_\infty(r;s)\psi\right\Vert 
\end{split} \end{equation}
where $\cK = d\Gamma (-\Delta)$ is the kinetic energy operator and where, in the last inequality, we used Cauchy-Schwarz. {F}rom Lemma \ref{lem:U2}, we find
\[ \begin{split}
 \left\Vert (\mathcal{N}+\right. & \left.1)^{j/2}(1 +\mathcal{N}/N)^{1/2}\; \mathcal{L}_4 \;\mathcal{U}_\infty(r;s)\psi\right\Vert\\
     \leq& \frac{C}{\sqrt{N}}e^{K|r-s|}\left(
  \left\Vert (\mathcal{N}+1)^{(j+3)/2}\psi\right\Vert
 +\frac{1}{N}\left\Vert (\mathcal{N}+1)^{(j+5)/2}(1+\mathcal{N}/N)\psi\right\Vert\right)\\
  & +\frac{C}{\sqrt{N}}e^{K|r-s|}\left(\left\Vert \mathcal{K}\psi\right\Vert
 +\frac{1}{N}\left\Vert (\mathcal{N}+1)^{(j+3)}(1+\mathcal{N}/N)\mathcal{U}_\infty(r;s)\psi\right\Vert\right)\\
   \leq & \frac{C}{\sqrt{N}}e^{K|r-s|}\left(\left\Vert (\mathcal{N}+1)^{(j+3)/2}\psi\right\Vert+\left\Vert \mathcal{K}\psi\right\Vert
 +\frac{1}{N}\left\Vert (\mathcal{N}+1)^{(j+3)}(1+\mathcal{N}/N)\psi\right\Vert\right).
 \end{split}
\]
Inserting the last equation and (\ref{eq:difference2}) into the r.h.s. of (\ref{eq:difference1}), we obtain the desired bound. 
\end{proof}

We will also need to control the growth of $\cN$ and of its power with respect to the unitary groups generated by operators of the form $h = N^{-1/2} d\Gamma (J) + \phi (f)$, where $J$ is a bounded operator on $L^2 (\bR^3)$ and $f \in L^2 (\bR^3)$. 
\begin{prop}\label{prop:h-bd}
Let $f \in L^2 (\bR^3)$ and $B$ be a bounded operator on $L^2 (\bR^3)$. Let 
\[ h = \frac{1}{\sqrt{N}} \, d\Gamma (B) + \phi (f) \]
where $\phi (f) = a(f) + a^* (f)$. For every $j \in \bN$ there exists a constant $C$ such that
\[ \langle \psi, e^{is h} (\cN+\alpha)^j e^{-is h} \psi \rangle \leq C \langle \psi, (\cN+\alpha+s^2 \| f \|^2)^j \psi \rangle  \]
for every $s \in \bR$, $\alpha \geq 1$. 
\end{prop}

{\it Remark:} from Proposition \ref{prop:h-bd} we obtain a bound for the norm
\begin{equation}\label{eq:norms-h} 
\| (\cN+\alpha)^{j/2} e^{-is h} (\cN+\alpha + s^2 \| f \|^2)^{-j/2} \| \leq C\;.
\end{equation}

\begin{proof}
We compute the derivative
\[ \begin{split} 
-i \frac{d}{ds} \, \langle \psi, e^{i hs} (\cN+\alpha)^j e^{-i hs} \psi \rangle = \; & \langle \psi, e^{ih s} [h, (\cN+\alpha)^j] e^{-i hs} \psi \rangle \\ = \; & \langle \psi, e^{ihs} [ \phi (f) , (\cN+\alpha)^j ] e^{-i hs} \psi \rangle \\
=\; & \sum_{\ell=0}^{j-1} \langle \psi, e^{ihs} (\cN+\alpha)^\ell a^* (f) (\cN+\alpha)^{j-1-\ell} e^{-ihs} \psi \rangle \\ &- \sum_{\ell=0}^{j-1} \langle \psi, e^{ihs} (\cN+\alpha)^\ell a (f) (\cN+\alpha)^{j-1-\ell} e^{-ihs} \psi \rangle\;.
\end{split} \]
We use the intertwining formulas $\cN a^* (f) = a^* (f) (\cN+1)$ and $(\cN+1) a(f) = a(f) \cN$ to write
\[ \begin{split} 
-i \frac{d}{ds} \, \langle \psi, &e^{i hs} (\cN+\alpha)^j e^{-i hs} \psi \rangle \\ 
=\; & \sum_{\ell=0}^{j-1} \langle \psi, e^{ihs}  (\cN+\alpha)^{j/2-1/4} a^* (f) (\cN+\alpha +1)^{\ell-j/2+1/4} (\cN+\alpha)^{j-1-\ell} e^{-ihs} \psi \rangle \\ &- \sum_{\ell=0}^{j-1} \langle \psi, e^{ihs} (\cN+\alpha)^\ell (\cN+\alpha+1)^{j/2-\ell-3/4} a (f) (\cN+\alpha)^{j/2-1/4} e^{-ihs} \psi \rangle\;.
\end{split} \]
Using the bounds (\ref{eq:est-a}), we find
\[ \begin{split} 
\Big| \frac{d}{ds} \langle &\psi, e^{ihs} (\cN+\alpha)^j e^{-ihs} \psi \rangle \Big| \\ \leq \; & \sum_{\ell=0}^{j-1} \| f \| \, \| (\cN+\alpha)^{j/2-1/4} e^{-ihs} \psi \|  \, \| (\cN+\alpha+1)^{\ell-j/2+3/4} (\cN+\alpha)^{j-\ell-1} e^{-ihs} \psi \| \\
\leq \; & C \| f \| \| (\cN+\alpha)^{j/2-1/4} e^{-ihs} \psi \|^2 \\
\leq \; &C\| f \| \langle \psi, e^{ihs} (\cN+\alpha)^j e^{-ihs} \psi \rangle^{1-1/2j}
\end{split} \]
for all $\alpha \geq 1$ and for a constant $C$ depending only on $j$. Gronwall's lemma gives
\[  \langle \psi, e^{ihs} (\cN+\alpha)^j e^{-ihs} \psi \rangle \leq C \| f \|^{2j} s^{2j} + \langle \psi, (\cN+\alpha)^j \psi \rangle \leq C \langle \psi,  (\cN+\alpha + \| f \|^2 s^2)^j \psi \rangle\]
for all $s \in \bR$.
\end{proof}

Finally, we need bounds on the growth of $\cN$, of its higher powers, and of $\cK$ with respect to the unitary group generated by self-adjoint field operators of the form $\phi (f)$, obtained from the operator $h$ introduced in Proposition \ref{prop:h-bd} in the limit $N \to \infty$.
\begin{lem}
Let $f \in L^2 (\bR^3)$ and $\phi (f) = a^* (f) + a(f)$. For every $j \in \bN$ there exists a constant $C$ such that
\[ \langle \psi, e^{is \phi (f)} (\cN+\alpha)^j e^{-is \phi (f)} \psi \rangle \leq C \langle \psi, (\cN+\alpha+s^2 \| f \|^2 )^j \psi \rangle \]
for all $s \in \bR$ and $\alpha \geq 1$. If $f \in H^1 (\bR^3)$, we have, for every $s \in \bR$ and $\alpha \geq 0$, 
\begin{equation}\label{eq:K-phi} 
\langle \psi, e^{i s\phi (f)} (\cK + \alpha) \, e^{-i s \phi (f)} \psi \rangle \leq 2 \left\langle \psi, (\cK + \alpha + s^2 \| \nabla f \|^2) \psi \right\rangle\;. \end{equation}
If $f \in H^2 (\bR^3)$ there exists a constant $C>0$ such that
\begin{equation}\label{eq:K2-phi} \langle \psi, e^{is \phi (f)} (\cK+\alpha)^2 e^{-is \phi (f)} \psi \rangle \leq C \left\langle \psi, (\cK + \alpha + s^2 \| \nabla f \|^2 + |s| \| \Delta f \|)^2 \psi \right\rangle \end{equation}
for every $s \in \bR$, $\alpha \geq 0$.
\end{lem}
{\it Remark:} from the lemma we obtain bounds for the norms 
\begin{equation}\label{eq:norms-phi} \begin{split} 
\| (\cN+\alpha)^{j/2} e^{-is \phi (f)} (\cN+\alpha + s^2 \| f \|^2)^{-j/2} \| &\leq C\;, \\
\| (\cK+\alpha)^{1/2} e^{-is \phi (f)} (\cK + \alpha + s^2 \| \nabla f \|^2)^{-1/2} \| & \leq C\;, \\
\| (\cK+\alpha) e^{-is\phi (f)} (\cK+\alpha +s^2 \| \nabla f \|^2 + |s| \| \Delta f \|)^{-1} \| & \leq C\;.
\end{split} \end{equation}

\begin{proof}
The first statement follows from Proposition \ref{prop:h-bd} taking $B=0$. To prove (\ref{eq:K-phi}), we observe that $e^{i\phi (f)} = W(if)$ is a Weyl operator. Therefore, we have
\[ \begin{split} e^{is \phi (f)} (\cK +\alpha) e^{-i s \phi (f)} &= \alpha + \int dx \nabla_x (a_x^* +i s \overline{f} (x)) \nabla_x (a_x -i s f (x)) 
\\ &= \cK + \alpha - is \int dx \nabla f (x) \cdot \nabla_x a_x^* + i s \int dx \nabla \overline{f} (x) \cdot \nabla_x a_x + s^2 \| \nabla f \|^2 \\ & \leq 2 (\cK + \alpha+ s^2 \| \nabla f \|^2) \end{split} \] 
which proves (\ref{eq:K-phi}). 

Finally, we show (\ref{eq:K2-phi}). We have
\[ \begin{split}
\langle \psi, e^{is\phi (f)} (\cK+\alpha)^2 e^{-is \phi (f)} \psi \rangle = \left\langle \psi, \left( \cK+ \alpha+ A +A^* + 
s^2 \| \nabla f \|^2 \right)^2 \psi \right\rangle 
\end{split} \]
with 
\[ A = -is \int dx \, \nabla \overline{f} (x) \cdot \nabla_x a_x\;. \]
By Cauchy-Schwarz, we find
\[ \langle \psi, e^{is\phi(f)} (\cK+\alpha)^2 e^{-is \phi (f)} \psi \rangle \leq C \langle \psi, ((\cK+\alpha)^2 + A^* A + A A^* + s^4 \| \nabla f \|^4) \psi \rangle\;. \]
We have
\[\begin{split}  
\langle \psi, A^* A \psi \rangle &= s^2 \int dx dy \nabla f (x) \nabla \overline{f} (y) \langle \psi, 
\nabla_x a_x^* \nabla_y a_y \psi \rangle \\ &\leq s^2 \int dx dy \, |\nabla f (x)| |\nabla f (y)| \| 
\nabla_x a_x \psi \| \, \| \nabla_y a_y \psi \| \\ &\leq s^2 \| \nabla f \|^2 \langle \psi, \cK \psi \rangle\;.
\end{split} \]
Since $[A,A^*] = s^2 \| \Delta f \|^2$, we conclude that
\[ \langle \psi, AA^* \psi \rangle \leq s^2 \| \Delta f \|^2 + s^2 \| \nabla f \|^2 \langle \psi, \cK \psi 
\rangle \]
and thus that
\[  \langle \psi, e^{is\phi(f)} (\cK+\alpha)^2 e^{-is \phi (f)} \psi \rangle \leq C \langle \psi, (\cK+\alpha +s^2 \| \nabla f \|^2 + |s| \| \Delta f \| )^2 \psi \rangle\;. 
\]
\end{proof}



\section{Proof of Theorem \ref{thm:multi-CLT}}\label{sect:main}
\setcounter{equation}{0}

To compute the expectation
\[ \E_{\psi_{N,t}} \big[ f_1(\cO_{1,t}) \dots f_k (\cO_{k,t}) \big]  = \langle \psi_{N,t} , f_1(\cO_{1,t}) \dots f_k (\cO_{k,t}) \psi_{N,t} \rangle \]
we expand the functions $f_1, \dots, f_k$ in their Fourier representation. We find
\[  \E_{\psi_{N,t}} \big[ f_1(\cO_{1,t}) \dots f_k (\cO_{k,t}) \big]  = \int d\tau_1 \dots d\tau_k \, \widehat{f}_1 (\tau_1) \dots \widehat{f}_k (\tau_k) \, \langle \psi_{N,t} , e^{i \tau_1 \cO_{1,t}} \dots e^{i\tau_k \cO_{k,t}} \psi_{N,t} \rangle\;. \]
Next, we embed our problem in the Fock-space. With a slight abuse of notation, we identify $\psi_{N,t}$ with the Fock space vector 
\[ \psi_{N,t} = e^{-i\cH_N t} \frac{a^* (\ph)^N}{\sqrt{N!}} \Omega \]
having only one non-zero component. We observe that 
\[ \psi_{N,t} = d_N P_N e^{-i\cH_N t} W(\sqrt{N} \ph) \Omega \]
where $P_N$ is the orthogonal projection onto the $N$-particle sector of the Fock space, and where $d_N = e^{N/2} N^{-N/2}\sqrt{N!} \simeq N^{1/4}$. We define
\[\wt{O}_{j,t} = O_j - \langle \ph_t , O_j \ph_t \rangle \, . \]
Since 
\[ \cO_{k,t} = \frac{1}{\sqrt{N}} d\Gamma (\wt{O}_{k,t}) |_{P_N \cF} \]
we find 
\begin{equation} \label{eq:pf1} \begin{split} 
\big\langle \psi_{N,t} , &e^{i \tau_1\cO_{1,t}} \dots e^{i \tau_k  \cO_{k,t}} \psi_{N,t} \big\rangle \\ = \; &d_N \left\langle  \frac{a^* (\ph)^N}{\sqrt{N!}} \Omega, e^{i\cH_N t} e^{i\frac{\tau_1}{\sqrt{N}} d\Gamma (\wt{O}_{1,t})} \dots e^{i\frac{\tau_k}{\sqrt{N}} d\Gamma (\wt{O}_{k,t})} e^{-i\cH_N t} P_N W(\sqrt{N} \ph) \Omega \right\rangle\\ = \; & d_N \left\langle W^* (\sqrt{N} \ph) \frac{a^* (\ph)^N}{\sqrt{N!}} \Omega , W^* (\sqrt{N} \ph) e^{i\cH_N t} e^{i \frac{\tau_1}{\sqrt{N}} d\Gamma (\wt{O}_{1,t})} \dots e^{i \frac{\tau_k}{\sqrt{N}} 
d\Gamma (\wt{O}_{k,t})} e^{-i\cH_N t} W(\sqrt{N} \ph) \Omega \right\rangle \\ 
= \; & \left\langle \xi_N, \cU_N^* (t;0) W^* (\sqrt{N} \ph_t) e^{i\frac{\tau_1}{\sqrt{N}} d\Gamma (\wt{O}_{1,t})} \dots e^{i\frac{\tau_k}{\sqrt{N}} d\Gamma (\wt{O}_{k,t})} W(\sqrt{N} \ph_t) \cU_N (t;0) \Omega \right\rangle 
\end{split} \end{equation}
where we introduced the fluctuation dynamics $\cU_N (t;s) = W^* (\sqrt{N} \ph_t) e^{-i\cH_N (t-s)} W(\sqrt{N} \ph_s)$ and where we defined the Fock space vector
\[ \xi_N = d_N W^* (\sqrt{N} \ph) \frac{a^* (\ph)^N}{\sqrt{N!}} \Omega\,. \]
Observe that $\| \xi_N \| = d_N \simeq N^{1/4}$. However, it follows from Lemma \ref{lm:xiN} that
\begin{equation}\label{eq:bd-xiN}
\| (\cN+1)^{-1/2} \xi_N \| \leq C 
\end{equation}
uniformly in $N$. {F}rom (\ref{eq:shifts}), we find
\[ \begin{split}  W^* (\sqrt{N} \ph_t) d\Gamma (\wt{O}_{j,t}) W(\sqrt{N} \ph_t) &= d\Gamma (\wt{O}_{j,t}) + \sqrt{N} \phi (\wt{O}_{j,t} \ph_t) + N \langle \ph_t, \wt{O}_{j,t} \ph_t \rangle \\ &= d\Gamma (\wt{O}_{j,t}) + \sqrt{N} \phi (\wt{O}_{j,t} \ph_t) \end{split} \]
because, by definition, $\langle \ph_t, \wt{O}_{j,t} \ph_t \rangle = 0$. Inserting in (\ref{eq:pf1}), we find 
\[ \begin{split} 
\big\langle \psi_{N,t} , &e^{i \tau_1\cO_{1,t}} \dots e^{i \tau_k  \cO_{k,t}} \, \psi_{N,t} \big\rangle = \left\langle \xi_N, \cU_N^* (t;0)  e^{i\tau_1 h_{1,t}} \dots e^{i\tau_k h_{k,t}} \cU_N (t;0) \Omega \right\rangle 
\end{split} \]
with
\[ h_{j,t} = \frac{1}{\sqrt{N}} d\Gamma (\wt{O}_{j,t}) + \phi (\wt{O}_{j,t} \ph_t) \]
for $j=1,\dots, k$. Recall here that $\phi (f) = a^* (f) + a(f)$, for any $f \in L^2 (\bR^3)$. We expand next $h_{j,t}$ around its main component $\phi (\wt{O}_{j,t} \ph_t)$. We find
\begin{equation}\label{eq:pf2} \begin{split}
\big\langle \psi_{N,t} , &e^{i \tau_1\cO_{1,t}} \dots e^{i \tau_k  \cO_{k,t}} \, \psi_{N,t} \big\rangle \\ = & \; \sum_{\ell=1}^k \left\langle \xi_N, \cU^*_N (t;0) \prod_{j=1}^{\ell-1} e^{i\tau_j  \phi (\wt{O}_{j,t} \ph_t)} \, \left(e^{i \tau_\ell h_{\ell,t}} - e^{i\tau_\ell \phi (\wt{O}_{j,t} \ph_t)} \right) \prod_{j=\ell+1}^k e^{i\tau_j h_{j,t}} \cU_N (t;0) \Omega \right\rangle \\ &+ \left\langle \xi_N, \cU^*_N (t;0) \prod_{j=1}^{k} e^{i\tau_j  \phi (\wt{O}_{j,t} \ph_t)} \,  \cU_N (t;0) \Omega \right\rangle\;.   \end{split} \end{equation}
In order to bound the terms in the sum over $\ell$, we write
\[ \begin{split} 
 \Big\langle \xi_N, & \, \cU^*_N (t;0) \prod_{j=1}^{\ell-1} e^{i\tau_j  \phi (\wt{O}_{j,t} \ph_t)} \, \left(e^{i \tau_\ell h_{\ell,t}} - e^{i\tau_\ell \phi (\wt{O}_{j,t} \ph_t)} \right) \prod_{j=\ell+1}^k e^{i\tau_j h_{j,t}} \cU_N (t;0) \Omega \Big\rangle \\
 &= \frac{i}{\sqrt{N}}  \int_0^{\tau_\ell} ds  \Big\langle \xi_N, \cU^*_N (t;0) \prod_{j=1}^{\ell-1} e^{i\tau_j  \phi (\wt{O}_{j,t} \ph_t)} \, e^{i(\tau_\ell-s) h_{\ell,t}} \\ & \hspace{6cm} \times d\Gamma (\wt{O}_{\ell,t}) e^{is \phi (\wt{O}_{\ell,t} \ph_t)}  \prod_{j=\ell+1}^k e^{i\tau_j h_{j,t}} \cU_N (t;0) \Omega \Big\rangle\;.
\end{split} \]
We estimate the absolute value of this term as follows:
\[ \begin{split} 
\Big|  \Big\langle \xi_N, & \, \cU^*_N (t;0) \prod_{j=1}^{\ell-1} e^{i\tau_j  \phi (\wt{O}_{j,t} \ph_t)} \, \left(e^{i \tau_\ell h_{\ell,t}} - e^{i\tau_\ell \phi (\wt{O}_{j,t} \ph_t)} \right) \prod_{j=\ell+1}^k e^{i\tau_j h_{j,t}} \cU_N (t;0) \Omega \Big\rangle \Big| \\
\leq \; & \frac{C}{\sqrt{N}} \int_0^{\tau_\ell} ds \, \| (\cN+1)^{-1/2} \xi_N \| \, \| (\cN+1)^{1/2} \cU_N (t;0) (\cN+1)^{-1} \| \\ & \hspace{1cm} \times \prod_{j=1}^{\ell-1} \| (\cN+1+\sum_{i=1}^{j-1} \tau_i^2 \| \wt{O}_{i,t} \ph_t \|^2) e^{i\tau_j \phi (\wt{O}_{j,t} \ph_t)} (\cN+1+ \sum_{i=1}^j \tau_i^2 \| \wt{O}_{i,t} \ph_t \|^2)^{-1} \|  \\
&\hspace{1cm} \times \|(\cN+1+ \sum_{i=1}^{\ell-1} \tau_i^2 \| \wt{O}_{i,t} \ph_t \|^2) e^{i (\tau_\ell -s) h_{\ell,s}}  (\cN+1+ \sum_{i=1}^{\ell} \tau_i^2 \| \wt{O}_{i,t} \ph_t \|^2 )^{-1} \| \\
&\hspace{1cm} \times  \|(\cN+1+ \sum_{i=1}^{\ell} \tau_i^2 \| \wt{O}_{i,t} \ph_t \|^2) d\Gamma (\wt{O}_{\ell,t})  (\cN+1+ \sum_{i=1}^{\ell} \tau_i^2 \| \wt{O}_{i,t} \ph_t \|^2)^{-2} \|
\\ &\hspace{1cm} \times 
 \| (\cN+1+ \sum_{i=1}^{\ell} \tau_i^2 \| \wt{O}_{i,t} \ph_t \|^2)^{2} e^{is \phi (\wt{O}_{\ell,t} \ph_t)} (\cN+1+ \sum_{i=1}^{\ell} \tau_i^2 \| \wt{O}_{i,t} \ph_t \|^2 )^{-2} \| \\
  &\hspace{1cm} \times  \prod_{j=\ell+1}^k \| (\cN+1+ \sum_{i=1}^{j-1} \tau_i^2 \| \wt{O}_{i,t} \ph_t \|^2)^{2} e^{i\tau_j h_{j,t}} (\cN+1+ \sum_{i=1}^{j} \tau_i^2 \| \wt{O}_{i,t} \ph_t \|^2 )^{-2}\| \\
   &\hspace{1cm} \times 
 \|(\cN+1+ \sum_{i=1}^{k} \tau_i^2 \| \wt{O}_{i,t} \ph_t \|^2)^{2} \, \cU_N (t;0)  \Omega \| \\ 
 \leq \;& \frac{Ce^{K|t|}}{\sqrt{N}} |\tau_\ell| \| \wt{O}_{\ell,t} \| \left(1 + \sum_{i=1}^k \tau_i^2 \| \wt{O}_{i,t} \ph_t \|^2 \right)^2\;.
 \end{split} \]
Here we used (\ref{eq:bd-xiN}), the norm bounds (\ref{eq:norms-h}) and (\ref{eq:norms-phi}), Proposition \ref{lem:UN} and the estimate
\[ \big\| d\Gamma (\wt{O}_{\ell,t}) (\cN+1+\sum_{i=1}^\ell \tau_i^2 \| \wt{O}_{i,t} \ph_t \|^2)^{-1} \big\| \leq \| \wt{O}_{\ell,t} \|\;. \]
{F}rom (\ref{eq:pf2}), we conclude that
\begin{equation}\label{eq:h-phi} \begin{split}  \big| \big\langle \psi_{N,t} , &e^{i \tau_1\cO_{1,t}} \dots e^{i \tau_k  \cO_{k,t}} \, \psi_{N,t} \big\rangle -\big\langle \xi_N, \cU^*_N (t;0) \prod_{j=1}^{k} e^{i\tau_j  \phi (\wt{O}_{j,t} \ph_t)} \,  \cU_N (t;0) \Omega \big\rangle \big| \\ &\hspace{5cm} \leq \frac{C e^{K |t|}}{\sqrt{N}} \left(\sum_{\ell=1}^k |\tau_\ell| \| \wt{O}_{\ell,t} \| \right) \left(1 + \sum_{i=1}^k \tau_i^2 \| \wt{O}_{i,t} \ph_t \|^2 \right)^{2}\;. \end{split} \end{equation}

Next, we replace the fluctuation dynamics $\cU_N (t;0)$ with the limiting dynamics $\cU_\infty (t;0)$ introduced in Proposition \ref{prop:GV}. To this end, we write
\begin{equation}\label{eq:UUin} \begin{split} 
\big\langle \xi_N, \cU^*_N (t;0) & \prod_{j=1}^{k} e^{i\tau_j  \phi (\wt{O}_{j,t} \ph_t)} \,  \cU_N (t;0) \Omega \big\rangle \\ = \; &\big\langle \xi_N, \cU^*_N (t;0)  \prod_{j=1}^{k} e^{i\tau_j  \phi (\wt{O}_{j,t} \ph_t)} \,  (\cU_N (t;0) - \cU_\infty (t;0)) \Omega \big\rangle \\
&+ \big\langle \xi_N, (\cU^*_N (t;0) - \cU^*_\infty (t;0)) \prod_{j=1}^{k} e^{i\tau_j  \phi (\wt{O}_{j,t} \ph_t)} \, \cU_\infty (t;0) \Omega \big\rangle \\
&+\big\langle \xi_N, \cU^*_\infty (t;0)  \prod_{j=1}^{k} e^{i\tau_j  \phi (\wt{O}_{j,t} \ph_t)} \, \cU_\infty (t;0) \Omega \big\rangle\;. \end{split} \end{equation}
To bound the first term on the r.h.s. of the last equation, we notice that
\begin{equation}\label{eq:UUin-1}\begin{split} 
\big| \big\langle \xi_N, & \cU^*_N (t;0)  \prod_{j=1}^{k} e^{i\tau_j  \phi (\wt{O}_{j,t} \ph_t)} \,  (\cU_N (t;0) - \cU_\infty (t;0)) \Omega \big\rangle \big| \\
\leq \; &\| (\cN+1)^{-1/2} \xi_N \| \| (\cN+1)^{1/2} \cU_N^* (t;0) (\cN+1)^{-1} \| \\
&\hspace{.5cm} \times \prod_{j=1}^k \| (\cN+1+ \sum_{i=1}^{j-1} \tau_i^2 \| \wt{O}_{i,t} \ph_t \|^2) e^{i\tau_j \phi (\wt{O}_{j,t} \ph_t)} (\cN+1+ \sum_{i=1}^{j} \tau_i^2 \| \wt{O}_{i,t} \ph_t \|^2)^{-1} \| \\
&\hspace{.5cm} \times \| (\cN+1+ \sum_{i=1}^{k} \tau_i^2 \| \wt{O}_{i,t} \ph_t \|^2) (\cU_N (t;0) - \cU_\infty (t;0)) \Omega \| \\
\leq \; &\frac{C e^{K|t|}}{\sqrt{N}}  \left(1+ \sum_{i=1}^k \tau_i^2 \| \wt{O}_{i,t} \ph_t \|^2 \right) \end{split} \end{equation}
where we used (\ref{eq:bd-xiN}), the norm bound (\ref{eq:norms-phi}), Proposition \ref{lem:UN} and Proposition \ref{prop:U-Uin}.  Next, we estimate the second term on the r.h.s. of (\ref{eq:UUin}). {F}rom Proposition \ref{prop:U-Uin} we find
\begin{equation}\label{eq:pf5} \begin{split} 
\big| \big\langle \xi_N, & (\cU^*_N (t;0) - \cU^*_\infty (t;0))  \prod_{j=1}^{k} e^{i\tau_j  \phi (\wt{O}_{j,t} \ph_t)} \,  \cU_\infty (t;0) \Omega \big\rangle \big| \\
\leq \; & \| (\cN+1)^{-1/2} \xi_N \| \, \| (\cN+1)^{1/2} (\cU^*_N (t;0) - \cU^*_\infty (t;0))  \prod_{j=1}^{k} e^{i\tau_j  \phi (\wt{O}_{j,t} \ph_t)} \,  \cU_\infty (t;0) \Omega \| \\
\leq \; &\frac{C e^{K|t|}}{\sqrt{N}} \, \| (\cN +1)^2 \prod_{j=1}^{k} e^{i\tau_j  \phi (\wt{O}_{j,t} \ph_t)} \, \cU_\infty (t;0) \Omega \| +\frac{Ce^{K|t|}}{\sqrt{N}} \| \cK \prod_{j=1}^{k} e^{i\tau_j  \phi (\wt{O}_{j,t} \ph_t)} \,  \cU_\infty (t;0) \Omega \| \\ &+ \frac{Ce^{K|t|}}{N^{3/2}} \| (\cN+1)^4 \prod_{j=1}^{k} e^{i\tau_j  \phi (\wt{O}_{j,t} \ph_t)} \, \cU_\infty (t;0) \Omega \| \\ &+ \frac{Ce^{K|t|}}{N^{5/2}} \| (\cN+1)^5 \prod_{j=1}^{k} e^{i\tau_j  \phi (\wt{O}_{j,t} \ph_t)} \,  \cU_\infty (t;0) \Omega \|\;. \end{split} \end{equation}
Using again (\ref{eq:norms-phi}), the first term on the r.h.s. of (\ref{eq:pf5}) can bounded by
\[ \begin{split} \frac{C e^{K|t|}}{\sqrt{N}} \, &\| (\cN +1)^2 \prod_{j=1}^{k} e^{i\tau_j  \phi (\wt{O}_{j,t} \ph_t)} \, \cU_\infty (t;0) \Omega \| \\  \leq \; & \frac{C e^{K|t|}}{\sqrt{N}} \prod_{j=1}^k \| (\cN+1 + \sum_{i=1}^{j-1} \tau_i^2 \| \wt{O}_{i,t} \ph_t \|^2)^2 e^{i\tau_j \phi (\wt{O}_{j,t} \ph_t)} (\cN+1 + \sum_{i=1}^{j} \tau_i^2 \| \wt{O}_{i,t} \ph_t \|^2)^{-2} \| \\ 
&\hspace{1cm} \times \| (\cN+1 + \sum_{i=1}^{k} \tau_i^2 \| \wt{O}_{i,t} \ph_t \|^2)^2 \cU_\infty (t;0) \Omega \| \\ \leq \; & \frac{C e^{K|t|}}{\sqrt{N}} \left( 1+ \sum_{i=1}^k \tau_i^2 \| \wt{O}_{i,t} \ph_t \|^2 \right)^2\;. \end{split} \]
Similarly, one can bound the third and the fourth term on the r.h.s. of (\ref{eq:pf5}). We obtain
\[  \frac{C e^{K|t|}}{N^{3/2}} \, \| (\cN +1)^4 \prod_{j=1}^{k} e^{i\tau_j  \phi (\wt{O}_{j,t} \ph_t)} \, \cU_\infty (t;0) \Omega \| \leq \frac{C e^{K|t|}}{N^{3/2}} \left( 1+ \sum_{i=1}^k \tau_i^2 \| \wt{O}_{i,t} \ph_t \|^2 \right)^4 \]
and
\[  \frac{C e^{K|t|}}{N^{5/2}} \, \| (\cN +1)^5 \prod_{j=1}^{k} e^{i\tau_j  \phi (\wt{O}_{j,t} \ph_t)} \, \cU_\infty (t;0) \Omega \| \leq \frac{C e^{K|t|}}{N^{5/2}} \left( 1+ \sum_{i=1}^k \tau_i^2 \| \wt{O}_{i,t} \ph_t \|^2 \right)^5\;. \]
To estimate the second term on the r.h.s of (\ref{eq:pf5}), we use the third bound in (\ref{eq:norms-phi}); we obtain
\[ \begin{split} \frac{C e^{K|t|}}{\sqrt{N}} \, &\| \cK \prod_{j=1}^{k} e^{i\tau_j  \phi (\wt{O}_{j,t} \ph_t)} \, \cU_\infty (t;0) \Omega \| \\  \leq \; & \frac{C e^{K|t|}}{\sqrt{N}} \prod_{j=1}^k \big\| \big(\cK + \sum_{i=1}^{j-1} (\tau_i^2 \| \nabla \wt{O}_{i,t} \ph_t \|^2 + |\tau_i| \| \Delta \wt{O}_{i,t} \ph_t \|) \big) e^{i\tau_j \phi (\wt{O}_{j,t} \ph_t)}\\ &\hspace{4cm} \times \big( \cK + \sum_{i=1}^{j} (\tau_i^2 \| \nabla \wt{O}_{i,t} \ph_t \|^2 + |\tau_j| \| \Delta \wt{O}_{i,t} \ph_t \|) \big)^{-1} \big\| \\ 
&\hspace{1cm} \times \| \big(\cK + \sum_{i=1}^{k} (\tau_i^2 \| \nabla \wt{O}_{i,t} \ph_t \|^2 + |\tau_i| \| \Delta \wt{O}_{i,t} \ph_t \|) \big) \cU_\infty (t;0) \Omega \| \\ \leq \; & \frac{C e^{K|t|}}{\sqrt{N}} \left( 1+ \sum_{i=1}^k (\tau_i^2 \| \nabla \wt{O}_{i,t} \ph_t \|^2 + |\tau_i| \| \Delta \wt{O}_{i,t} \ph_t \|) \right)\;. \end{split} \]
Combining the last four bounds, we conclude that
\[ \begin{split} \big| \big\langle \xi_N, & (\cU^*_N (t;0) - \cU^*_\infty (t;0))  \prod_{j=1}^{k} e^{i\tau_j  \phi (\wt{O}_{j,t} \ph_t)} \,  \cU_\infty (t;0) \Omega \big\rangle \big| \\
\leq \; & \frac{C e^{K|t|}}{\sqrt{N}} \left(1 + \left(\sum_{i=1}^k \tau_i^2 \| \wt{O}_{i,t} \ph_t \|^2\right)^2  + \sum_{i=1}^k (\tau_i^2 \| \nabla \wt{O}_{i,t} \ph_t \|^2 + |\tau_i| \| \Delta \wt{O}_{i,t} \ph_t \|) \right) \\ &+  \frac{C e^{K|t|}}{N^{3/2}} \left( \sum_{i=1}^k \tau_i^2 \| \wt{O}_{i,t} \ph_t \|^2 \right)^4 +  \frac{C e^{K|t|}}{N^{5/2}} \left( \sum_{i=1}^k \tau_i^2 \| \wt{O}_{i,t} \ph_t \|^2 \right)^5.
\end{split} \] 
Combining with (\ref{eq:UUin-1}) and (\ref{eq:UUin}), we get from (\ref{eq:h-phi}) that
\begin{equation}\label{eq:h0Uin} \begin{split}  \big| \big\langle \psi_{N,t} , e^{i \tau_1\cO_{1,t}} \dots e^{i \tau_k  \cO_{k,t}} \, &\psi_{N,t} \big\rangle -\big\langle \xi_N, \cU^*_\infty (t;0) \prod_{j=1}^{k} e^{i\tau_j  \phi (\wt{O}_{j,t} \ph_t)} \,  \cU_\infty (t;0) \Omega \big\rangle \big| \\  
\leq \; & \frac{C e^{K |t|}}{\sqrt{N}} \left(1 + \sum_{\ell=1}^k |\tau_\ell|^5 \| \wt{O}_{\ell,t} \|^5 \right) \\ 
&+\frac{C e^{K|t|}}{\sqrt{N}} \sum_{i=1}^k (\tau_i^2 \| \nabla \wt{O}_{i,t} \ph_t \|^2 + |\tau_i| \| \Delta \wt{O}_{i,t} \ph_t \|) \\ &+  \frac{C e^{K|t|}}{N^{3/2}} \sum_{i=1}^k \tau_i^8 \| \wt{O}_{i,t} \ph_t \|^8 +  \frac{C e^{K|t|}}{N^{5/2}} \sum_{i=1}^k \tau_i^{10} \| \wt{O}_{i,t} \ph_t \|^{10}.
\end{split} \end{equation}

Now we want to replace the $N$-dependent vector $\xi_N$ with its limit. This procedure will produce again an error of size $N^{-1/2}$. {F}rom Lemma \ref{lm:xiN}, we have
\[ \xi_N = d_N W^* (\sqrt{N} \ph) \frac{a^* (\ph)^N}{\sqrt{N!}} \Omega = \sum_{\ell=0}^\infty \xi_N^{(\ell)} a^* (\ph)^\ell \Omega \]
where the coefficients $\xi_N^{(\ell)}$ satisfy the recursion 
\begin{equation}\label{eq:rec1} \xi_N^{(\ell)} = \frac{1-\ell}{\ell} N^{-1/2} \xi_N^{(\ell-1)} - \frac{1}{\ell} \xi_N^{(\ell-2)} \end{equation}
with the initial data $\xi_N^{(0)} = 1$ and $\xi_N^{(1)} = 0$. Recall that $\| \xi_N \| = d_N \simeq N^{1/4}$ and that, from Lemma \ref{lm:xiN}, 
\begin{equation}\label{eq:apri} \| (\cN+1)^{-1/2} \xi_N \|^2 = \sum_{\ell \geq 0} \frac{\ell !}{\ell+1} |\xi_N^{(\ell)}|^2 \leq C \end{equation}
uniformly in $N$ (the factor $\ell !$ arises because $\| a^* (\ph)^\ell \Omega \|^2 = \ell !$). 

We compare the coefficients $\xi_N^{(\ell)}$ with the limiting coefficients 
\begin{equation}\label{eq:xiinf} \xi_\infty^{(\ell)} = \left\{ \begin{array}{ll} 0 \quad &\text{if } \ell = 2m+1 \\ \frac{(-1)^m}{2^m m!} \quad &\text{if } \ell = 2m \end{array} \right. \end{equation}
which satisfy the recursion
\begin{equation}\label{eq:rec2} \xi_\infty^{(\ell)} = - \frac{1}{\ell} \xi_\infty^{(\ell-2)} \end{equation}
with the initial data $\xi_\infty^{(0)} = 1$ and $\xi_\infty^{(1)} = 0$. Observe here that $|\xi_\infty^{(\ell)}| \simeq \ell^{-1/4}$ (for $\ell$ even), and therefore  
\[ \sum_{\ell \geq 0} |\xi_\infty^{(\ell)}|^2 = \infty\;. \]
This means that the vector 
\[ \xi_\infty = \sum_{\ell \geq 0} \xi_\infty^{(\ell)} a^* (\ph)^\ell \Omega \]
is not an element of the Fock space (this is not surprising since $\| \xi_N \| \simeq N^{1/4} \to \infty$, as $N \to \infty$). We will avoid this problem by considering the vector
\begin{equation}\label{eq:xiin} (\cN+1)^{-\alpha} \xi_\infty = \sum_{\ell \geq 0} (\ell+1)^{-\alpha} \xi_\infty^{(\ell)} a^* (\ph)^\ell \Omega \end{equation}
which is in the Fock space for $\alpha > 1/2$, and showing that 
\[ \| (\cN+1)^{-\alpha} \xi_N - (\cN+1)^{-\alpha} \xi_\infty \| \leq C N^{-1/2} \]
One should think of $(\cN+1)^{-\alpha} \xi_\infty$ as a notation for the right hand side of (\ref{eq:xiin}), and not for the action of the operator $(\cN+1)^{-\alpha}$ on the (non-existing) vector $\xi_\infty$. 

Comparing the recursions (\ref{eq:rec1}) and (\ref{eq:rec2}), we obtain
\[ \xi_N^{(\ell)} - \xi_\infty^{(\ell)} = \frac{1-\ell}{\ell} N^{-1/2} \xi_N^{(\ell-1)} - \frac{1}{\ell} (\xi_N^{(\ell-2)} - \xi_\infty^{(\ell-2)})\;. \]
Taking absolute value, we find
\[ | \xi_N^{(\ell)} - \xi_\infty^{(\ell)}| \leq \frac{1}{\ell} |\xi_N^{(\ell-2)} - \xi_\infty^{(\ell-2)}| + \frac{1}{\sqrt{N}} |\xi_N^{(\ell-1)}| \, . \]
This implies that
\[ \sqrt{\ell !} \, |\xi_N^{(\ell)} - \xi_\infty^{(\ell)}| \leq   \sqrt{\frac{\ell-1}{\ell}} \, \sqrt{(\ell-2)!} |\xi_N^{(\ell-2)} - \xi_\infty^{(\ell-2)}| + \frac{\sqrt{\ell!}}{\sqrt{N}} |\xi_N^{(\ell-1)}| \]
and that, for every $\alpha \in \bN$,
\[ \frac{\sqrt{\ell !}}{(\ell+1)^{\alpha/2}}  \, |\xi_N^{(\ell)} - \xi_\infty^{(\ell)}| \leq   \sqrt{\frac{(\ell-1)^{\alpha+1}}{\ell (\ell+1)^\alpha}} \, \frac{\sqrt{(\ell-2)!}}{(\ell-1)^{\alpha/2}} |\xi_N^{(\ell-2)} - \xi_\infty^{(\ell-2)}| + \frac{1}{\sqrt{N}} \frac{\sqrt{\ell!}}{(\ell+1)^{\alpha/2}} |\xi_N^{(\ell-1)}|\;. \]
Next, we take the square, using that $(a+b)^2 \leq a^2 (1+ \eps) + b^2 (1+\eps^{-1})$, for every $\eps, a, b >0$. We find
\[ \frac{\ell !}{(\ell+1)^{\alpha}}  \, |\xi_N^{(\ell)} - \xi_\infty^{(\ell)}|^2 \leq \frac{(\ell-1)^{\alpha+1}}{\ell (\ell+1)^\alpha} \, (1+\eps) \, \frac{(\ell-2)!}{(\ell-1)^{\alpha}} |\xi_N^{(\ell-2)} - \xi_\infty^{(\ell-2)}|^2 + \frac{1}{N} (1+ \eps^{-1}) \, \frac{\ell!}{(\ell+1)^{\alpha}} |\xi_N^{(\ell-1)}|^2\;. \]
We choose $\eps > 0$ such that
\[ \frac{(\ell-1)^{\alpha+1}}{\ell (\ell+1)^\alpha} \, (1+\eps) = 1 \, \quad \Rightarrow \quad \eps  = \frac{\ell (\ell+1)^\alpha - (\ell-1)^{\alpha+1}}{(\ell-1)^{\alpha+1}}   \]
and find
\[ \frac{\ell !}{(\ell+1)^{\alpha}}  \, |\xi_N^{(\ell)} - \xi_\infty^{(\ell)}|^2 \leq \, \frac{(\ell-2)!}{(\ell-1)^{\alpha}} |\xi_N^{(\ell-2)} - \xi_\infty^{(\ell-2)}|^2 + \frac{1}{N} \frac{\ell^3}{\ell (\ell+1)^\alpha - (\ell-1)^{\alpha+1}} \, \frac{(\ell-1)!}{\ell} |\xi_N^{(\ell-1)}|^2\;. \]
We choose $\alpha = 3$, so that
\[ \frac{\ell^3}{\ell (\ell+1)^\alpha - (\ell-1)^{\alpha+1}} \leq C \]
for all $\ell \in \bN$. This gives
\[  \frac{\ell !}{(\ell+1)^3}  \, |\xi_N^{(\ell)} - \xi_\infty^{(\ell)}|^2 \leq \, \frac{(\ell-2)!}{(\ell-1)^3} |\xi_N^{(\ell-2)} - \xi_\infty^{(\ell-2)}|^2 + \frac{C}{N} \, \frac{(\ell-1)!}{\ell} |\xi_N^{(\ell-1)}|^2\;. \]
Iterating the last inequality, and using the fact that $|\xi_N^{(\ell)} - \xi_\infty^{(\ell)}| = 0$ for $\ell = 0,1$, we conclude that
\[ \frac{\ell!}{(\ell+1)^3} |\xi_N^{(\ell)} - \xi_\infty^{(\ell)} |^2 \leq \frac{C}{N} \sum_{j=1}^\ell \frac{(j-1)!}{j} |\xi_N^{(j)}|^2 \leq \frac{C}{N} \sum_{j=1}^\infty \frac{(j-1)!}{j} |\xi_N^{(j)}|^2 \leq \frac{C}{N} \]
uniformly in $\ell \in \bN$. Here we used (\ref{eq:apri}). Therefore, we find
\[ \| (\cN+1)^{-5/2} \xi_N - (\cN+1)^{-5/2} \xi_\infty \|^2 = \sum_{\ell \geq 0} \frac{\ell!}{(\ell+1)^5} |\xi_N^{(\ell)}|^2 \leq \frac{C}{N} \sum_{\ell \geq 0} \frac{1}{(\ell+1)^2} \leq \frac{C}{N}\;. \]
This implies that
\[ \begin{split} 
\Big| \langle \xi_N, \cU_\infty^* &(t;0) \prod_{j=1}^k e^{i\tau_j \phi (\wt{O}_{j,t} \ph_t)} \cU_\infty (t;0) \Omega \rangle - \langle \xi_\infty, \cU_\infty^* (t;0) \prod_{j=1}^k e^{i\tau_j \phi (\wt{O}_{j,t} \ph_t)} \cU_\infty (t;0) \Omega \rangle \Big| \\ \leq \; & \| (\cN+1)^{-5/2} \xi_N - (\cN+1)^{-5/2} \xi_\infty \|  \, \| (\cN+1)^{5/2} \cU_\infty^* (t;0) (\cN+1)^{-5/2} \| \\ & \hspace{.5cm} \times \prod_{j=1}^k \| (\cN+1+\sum_{i=1}^{j-1} \tau_i^2 \| \wt{O}_{i,t} \ph_t \|^2)^{5/2} e^{i\tau_j \phi (\wt{O}_{j,t} \ph_t)} (\cN+1+\sum_{i=1}^{j} \tau_i^2 \| \wt{O}_{i,t} \ph_t \|^2)^{-5/2} \| \\ 
&\hspace{.5cm} \times \| (\cN+1 +\sum_{i=1}^{k} \tau_i^2 \| \wt{O}_{i,t} \ph_t \|^2)^{5/2} \cU_\infty (t;0) \Omega \| \\ \leq \; &\frac{C e^{K|t|}}{\sqrt{N}} \left( 1 + \sum_{i=1}^k \tau_i^2 \| \wt{O}_{i,t} \ph_t \|^2 \right)^{5/2}\;. \end{split} \]
{F}rom (\ref{eq:h0Uin}), we obtain
\begin{equation}\label{eq:pf8a} \begin{split}
\Big| \langle \psi_{N,t}, e^{i\tau_1 \cO_{1,t}} \dots e^{i\tau_k \cO_{k,t}} \psi_{N,t} \rangle - 
&\langle \xi_\infty, \cU_\infty^* (t;0) \prod_{j=1}^k e^{i\tau_j \phi (\wt{O}_{j,t} \ph_t)} \cU_\infty (t;0) \Omega \rangle \Big| \\ 
\leq \; & \frac{C e^{K |t|}}{\sqrt{N}} \left(1 + \sum_{\ell=1}^k |\tau_\ell|^5 \| \wt{O}_{\ell,t}  \|^5 \right) \\ 
&+\frac{C e^{K|t|}}{\sqrt{N}} \sum_{i=1}^k (\tau_i^2 \| \nabla \wt{O}_{i,t} \ph_t \|^2 + |\tau_i| \| \Delta \wt{O}_{i,t} \ph_t \|) \\ &+  \frac{C e^{K|t|}}{N^{3/2}} \sum_{i=1}^k \tau_i^8 \| \wt{O}_{i,t} \ph_t \|^8 +  \frac{C e^{K|t|}}{N^{5/2}} \sum_{i=1}^k \tau_i^{10} \| \wt{O}_{i,t} \ph_t \|^{10}.
\end{split} \end{equation}
{F}rom the bounds on $\| O_j \|$, $\| \nabla O_j (1-\Delta)^{-1/2} \|$, $\| \Delta O_j (1-\Delta)^{-1} \|$ (which clearly imply bounds on $\| \wt{O}_j \|$, $\| \nabla \wt{O}_j (1-\Delta)^{-1/2} \|$, $\| \Delta \wt{O}_j (1-\Delta)^{-1} \|$) and from the bound on $\| \ph \|_{H^2}$ (which implies a bound on $\| \ph_t \|_{H^2} \leq C e^{K|t|} \| \ph \|_{H^2}$ for every $t \in \bR$), we conclude that
\begin{equation}\label{eq:pf8} \begin{split}
\Big| \langle \psi_{N,t}, e^{i\tau_1 \cO_{1,t}} \dots e^{i\tau_k \cO_{k,t}} \psi_{N,t} \rangle - 
&\langle \xi_\infty, \cU_\infty^* (t;0) \prod_{j=1}^k e^{i\tau_j \phi (\wt{O}_{j,t} \ph_t)} \cU_\infty (t;0) \Omega \rangle \Big| \\ 
\leq \; & \frac{C e^{K |t|}}{\sqrt{N}} \left(1 + \sum_{\ell=1}^k |\tau_\ell|^5 + \frac{1}{N} \sum_{i=1}^k \tau_i^8 + \frac{1}{N^2} \sum_{i=1}^k \tau_i^{10} \right)\;. 
\end{split} \end{equation}

Here, one should be careful with the notation. As explained above, $\xi_\infty$ is not a Fock space vector. In the last three equations, the Fock space inner product involving $\xi_\infty$ should be really understood as
\[ \langle (\cN+1)^{-5/2} \xi_\infty , (\cN+1)^{5/2} \cU_\infty^* (t;0) \prod_{j=1}^k e^{i\tau_j \phi (\wt{O}_{j,t} \ph_t)} \cU_\infty (t;0) \Omega \rangle \]
with $(\cN+1)^{-5/2} \xi_\infty$ defined as the Fock space vector (\ref{eq:xiin}). It is only to shorten the notation that we remove the two factors $(\cN+1)^{-5/2}$ and $(\cN+1)^{5/2}$.  

Next we notice that, with the notation $A(f,g) = a(f) + a^* (Jg)$ (where $Jg= \overline{g}$) introduced after Proposition \ref{prop:GV}, we have 
\[ \phi (\wt{O}_{j,t} \ph_t) = A(\wt{O}_{j,t} \ph_t, J \wt{O}_{j,t} \ph_t)\,. \]
Hence, by Proposition \ref{prop:bog},
\[ \begin{split} \cU_\infty^* (t;0) \phi (\wt{O}_{j,t} \ph_t) \cU_\infty (t;0) &= A(\Theta (t;0) (\wt{O}_{j,t} \ph_t, J \wt{O}_{j,t} \ph_t) ) \\ &= A(U (t;0) \wt{O}_{j,t} \ph_t + J V (t;0) \wt{O}_{j,t} \ph_t , J U (t;0) \wt{O}_{j,t} + V (t;0) \wt{O}_{j,t} \ph_t) \\ &= \phi (\wt{g}_{j,t}) \end{split} \]
where we defined
\[ \wt{g}_{j,t} = U (t;0) \wt{O}_{j,t} \ph_t + J V (t;0) \wt{O}_{j,t} \ph_t\,. \]
Here $\Theta (t;s)$ is the Bogoliubov transform defined in Proposition \ref{prop:bog}, which, according to (\ref{eq:bog-dec}), can be decomposed as
\[ \Theta (t;s) = \left( \begin{array}{ll} U (t;s) & J V(t;s) J \\ V(t;s) & J U (t;s) J \end{array} \right)\,. \]
It follows that
\begin{equation}\label{eq:pf6}  \langle \xi_\infty, \cU_\infty^* (t;0) \prod_{j=1}^k   e^{i\tau_j \phi (\wt{O}_{j,t} \ph_t)} \cU_\infty (t;0) \Omega \rangle = \langle \xi_\infty, \prod_{j=1}^k e^{i\tau_j \phi (\wt{g}_{j,t})} \Omega \rangle\,.
\end{equation}
For $f,g \in L^2 (\bR^3)$, the canonical commutation relations (\ref{eq:CCR}) imply that 
\[ [ \phi (f), \phi (g) ] = [ a(f) + a^* (f) , a(g) + a^* (g)]  = \langle f,g \rangle - \langle g,f \rangle  = 2i \text{Im } \langle f,g \rangle\,. \]
Hence, the Baker-Campbell-Hausdorff formula implies that
\[ e^{i \phi (f)} e^{i\phi (g)} = e^{i \phi (f+g)} e^{-i \text{Im } \langle f,g \rangle} \,.\]
We obtain
\[ \prod_{j=1}^k e^{i\tau_j \phi (\wt{g}_{j,t})} = e^{i \phi (\tau_1 \wt{g}_{1,t} + \dots + \tau_k \wt{g}_{k,t})} \, \prod_{i<j}^k e^{-i \tau_i \tau_j \text{Im } \langle \wt{g}_{i,t}, \wt{g}_{j,t} \rangle}\;. \]
Let $\wt{g} = \sum_{j=1}^k \tau_j \wt{g}_{j,t}$. Again from the Baker-Hausdorff-Campbell formula we find
\[ e^{i \phi (\wt{g})} \Omega = e^{-\| \wt{g} \|^2/2} \sum_{m \geq 0} i^m \frac{a^* (\wt{g})^m}{m!} \Omega\;. \]
Therefore, from (\ref{eq:pf6}), we have
\[ \begin{split} \langle \xi_\infty, \cU_\infty^* (t;0) \prod_{j=1}^k  &e^{i\tau_j \phi (\wt{O}_{j,t} \ph_t)} \cU_\infty (t;0) \Omega \rangle \\ = \; & e^{-\| \wt{g} \|^2/2} \, \prod_{i<j}^k e^{-i \tau_i \tau_j \text{Im } \langle \wt{g}_{i,t}, \wt{g}_{j,t} \rangle} \,  \sum_{\ell \geq 0} i^\ell \frac{\xi_\infty^{(\ell)}}{\ell!}  \langle a^* (\ph)^\ell \Omega, a^* (\wt{g})^\ell \Omega \rangle \;.
\end{split} \]
{F}rom (\ref{eq:xiinf}), and since
\[ \langle a^* (\ph)^\ell \Omega, a^* (\wt{g})^\ell \Omega \rangle = \ell! \, \langle \ph , \wt{g} \rangle^\ell \]
we obtain 
\[ \begin{split} \langle \xi_\infty, \cU_\infty^* (t;0) \prod_{j=1}^k  &e^{i\tau_j \phi (\wt{O}_{j,t} \ph_t)} \cU_\infty (t;0) \Omega \rangle   \\ = \; & e^{-\frac{1}{2} \| \sum_{j=1}^k \tau_j \wt{g}_{j,t} \|^2} \, e^{\frac{1}{2} \langle \ph , \sum_{j=1}^k \tau_j \wt{g}_{j,t} \rangle^2} \prod_{i<j}^k e^{-i \tau_i \tau_j \text{Im } \langle \wt{g}_{i,t}, \wt{g}_{j,t} \rangle} \\
= \; & e^{-\frac{1}{2} \sum_{i,j =1}^k \Sigma_{ij} (t) \tau_i \tau_j} \end{split} \]
where we defined the $k \times k$ matrix $\Sigma (t) = (\Sigma_{ij} (t))$ through
\[ \Sigma_{ij} (t) = \langle \wt{g}_{i,t}, \wt{g}_{j,t} \rangle - \langle \ph , \wt{g}_{i,t} \rangle \langle \ph, \wt{g}_{j,t} \rangle \]
for all $i\leq j$ and through $\Sigma_{ij} (t) = \Sigma_{ji} (t)$ for $i > j$. We notice here that the factors $\langle \ph, \wt{g}_{i,t} \rangle$ and $\langle \ph, \wt{g}_{j,t} \rangle$ are real. In fact, for any self-adjoint operator $O$ on $L^2 (\bR^3)$, we have 
\[ \Theta (t;0) (O \ph_t , J O \ph_t) = (U(t;0) O \ph_t + J V(t;0) O \ph_t, J U (t;0) O \ph_t + V(t;0) O \ph_t) \]
and 
\[ \langle (\ph, -J\ph), \Theta (t;0) (O \ph_t , J O \ph_t) \rangle_{L^2 \oplus L^2} = 2i \text{Im } \langle \ph,  U(t;0) O \ph_t + J V(t;0) O \ph_t \rangle \, .  \]
On the other hand
\[ \begin{split} \langle  (\ph, -J\ph), \Theta (t;0) (O \ph_t , J O \ph_t) \rangle_{L^2 \oplus L^2} = \; & \langle \Theta^* (t;0) (\ph, - J \ph), (O \ph_t , J O \ph_t) \rangle_{L^2 \oplus L^2} \\ =\; &
\langle \Theta^{-1} (t;0) (\ph,J\ph), (O \ph_t , -J O \ph_t) \rangle_{L^2 \oplus L^2} \\ =
\; &\langle (\ph_t, J \ph_t), (O \ph_t , -J O \ph_t) \rangle_{L^2 \oplus L^2}\\ = \; & 2i \text{Im } \langle \ph_t, O \ph_t \rangle  = 0 \end{split} \]
where we used the relation $\Theta^* (t;0) = S \Theta^{-1} (t;0) S$, with $S$ defined in (\ref{eq:CCR-AA}).

We also notice that
\[ \begin{split} \langle \wt{g}_{i,t}, \wt{g}_{j,t} \rangle = \; &\langle U (t;0) \wt{O}_{i,t} \ph_t + J V(t;0) \wt{O}_{i,t} \ph_t, \wt{g}_{j,t} \rangle \\ =\; & \langle U (t;0) O_{i} \ph_t + J V(t;0) O_{i} \ph_t, \wt{g}_{j,t} \rangle - \langle \ph_t, O_i \ph_t \rangle \langle U(t;0) \ph_t + J V(t;0) \ph_t , \wt{g}_{j,t} \rangle \\ =\; &
 \langle U (t;0) O_{i} \ph_t + J V(t;0) O_{i} \ph_t,U (t;0) O_j \ph_t + J V(t;0) O_j \ph_t \rangle \\ &- \langle \ph_t, O_j \ph_t \rangle \langle U (t;0) O_{i} \ph_t + J V(t;0) O_{i} \ph_t, U(t;0) \ph_t + J V(t;0) \ph_t \rangle \\ & - \langle \ph_t, O_i \ph_t \rangle \langle U(t;0) \ph_t + J V(t;0) \ph_t , \wt{g}_{j,t} \rangle\;.\end{split} \]
Since, from Proposition \ref{prop:bog}, $ (\ph , J \ph) = \Theta (t;0) (\ph_t, J \ph_t) = (U (t;0) \ph_t + J V(t;0) \ph_t, J U(t;0) \ph_t + V(t;0) \ph_t)$, we see that $U(t;0) \ph_t + J V(t;0) \ph_t = \ph$, and therefore
\[  \begin{split} \langle \wt{g}_{i,t}, \wt{g}_{j,t} \rangle  =\; &
 \langle U (t;0) O_{i} \ph_t + J V(t;0) O_{i} \ph_t,U (t;0) O_j \ph_t + J V(t;0) O_j \ph_t \rangle \\ &- \langle \ph_t, O_j \ph_t \rangle \langle U (t;0) O_{i} \ph_t + J V(t;0) O_{i} \ph_t, \ph \rangle \\ & - \langle \ph_t, O_i \ph_t \rangle \langle \ph , \wt{g}_{j,t} \rangle\;. \end{split} \]
Similarly, we find
\[ \begin{split} \langle \ph , \wt{g}_{i,t}  \rangle \langle \ph , \wt{g}_{j,t}  \rangle = \; &\langle \ph , U(t;0) O_i\ph_t + J V(t;0) O_i \ph_t \rangle \langle \ph , U(t;0) O_j \ph_t + J V(t;0)  O_j \ph_t  \rangle 
\\ &-\langle \ph, U(t;0) O_i\ph_t + J V(t;0) O_i \ph_t  \rangle \langle \ph_t, O_j \ph_t \rangle  
- \langle \ph_t, O_i \ph_t \rangle \langle \ph , \wt{g}_{j,t}  \rangle\,. \end{split} \]
Hence
\[ \Sigma_{ij} (t) = \langle g_{i,t} , g_{j,t} \rangle - \langle \ph, g_{i,t} \rangle \langle \ph , g_{j,t} \rangle \]
with
\[ g_{i,t} = U(t;0) O_i \ph_t + J V(t;0) O_i \ph_t \]
and similarly for $g_{j,t}$ (the products $\langle \ph, g_{i,t} \rangle$ are, like $\langle \ph, \wt{g}_{i,t} \rangle$, real). 

We observe that the matrix $\Sigma (t)$ can be decomposed in its real and imaginary part $\Sigma (t) = P(t) + i R(t)$, with the two symmetric $k \times k$ matrices $P(t) = (P_{ij} (t))$ and $R(t) = (R_{ij} (t))$ given by  
\[ P_{ij} (t) = \text{Re } \langle g_{i,t}, g_{j,t} \rangle - \langle \ph, g_{i,t} \rangle \langle \ph, g_{j,t} \rangle \]
for all $i,j = 1, \dots ,k$, and 
\[ R_{ij} (t) = \text{Im } \langle g_{i,t}, g_{j,t} \rangle \]
for all $i < j$ and $R_{ij} (t) = R_{ji} (t)$ for all $i > j$. The real part $P(t)$ is always non-negative (see (\ref{eq:posit})). Under the assumption that $P(t)$ is strictly positive, $\Sigma (t)$ is invertible. We denote by $\Sigma^{-1} (t)$ its inverse. Then $\text{Re } \Sigma^{-1} (t) > 0$ and we have
\[ \int dx_1, \dots dx_k \, e^{i\sum_{j=1}^k x_j \tau_j} \left[ \frac{1}{\sqrt{(2\pi)^{k} \det \Sigma (t)}} e^{-\frac{1}{2} \sum_{i,j=1}^k \Sigma^{-1}_{ij} (t) x_i x_j} \right]= e^{-\frac{1}{2} \sum_{i,j=1}^k \Sigma_{ij} (t) \tau_i \tau_j}\;. \]
Hence
\[ \begin{split} 
\int d\tau_1 \dots d\tau_k \, \widehat{f}_1 (\tau_1) &\dots \widehat{f}_k (\tau_k) \langle \xi_\infty, \cU_\infty^* (t;0) \prod_{j=1}^k e^{i\tau_j \phi (\wt{O}_{j,t} \ph_t)} \cU_\infty (t;0) \Omega \rangle \\ = \; & \int d\tau_1 \dots d\tau_k \, \widehat{f}_1 (\tau_1) \dots \widehat{f}_k (\tau_k) e^{-\frac{1}{2} \sum_{i,j=1}^k \Sigma_{ij} (t) \tau_i \tau_j} \\ = \; &\int dx_1 \dots dx_k f_1 (x_1) \dots f_k (x_k) \, \left[ \frac{1}{ \sqrt{(2\pi)^k \det \Sigma (t)}} \, e^{-\frac{1}{2} \sum_{i,j=1}^k \Sigma^{-1}_{ij} (t) x_i x_j} \right] \;.\end{split} \]
{F}rom (\ref{eq:pf8}), we conclude that, under the assumptions of Theorem \ref{thm:multi-CLT}, 
\[ \begin{split} \Big| \langle \psi_{N,t} f_1 (\cO_{1,t}) \dots f_k (\cO_{k,t}) \psi_{N,t} \rangle - &\int dx_1 \dots dx_k \, f_1 (x_1) \dots f_k (x_k) \left[ \frac{e^{-\frac{1}{2} \sum_{i,j=1}^k \Sigma^{-1}_{ij} (t) x_i x_j}}{\sqrt{(2\pi)^k \det \Sigma (t)}} \right] \Big| \\ \leq \; &\frac{C e^{K|t|}}{\sqrt{N}} \prod_{j=1}^k \int d\tau |\widehat{f}_j (\tau)| (1+|\tau|^5 + N^{-1} \tau^8 + N^{-2} \tau^{10}) \,. \end{split} \]

\appendix

\section{Properties of $\xi_N$} 
\setcounter{equation}{0}

We collect some properties of the Fock space vector $\xi_N = d_N W^* (\sqrt{N} \ph) \ph^{\otimes N}$ which have been used in the proof of Theorem \ref{thm:multi-CLT}. The proof of the next lemma can be found in \cite{BKS}.

\begin{lem}\label{lm:xiN} 
For $\ph \in L^2 (\bR^3)$, set
\[ \xi_N = d_N W^* (\sqrt{N} \ph) \frac{a^*(\ph)^N}{\sqrt{N!}} \Omega \]
with $d_N = e^{N/2} \sqrt{N!} \, N^{-N/2} \simeq N^{1/4}$. Then there exists a constant $C > 0$ such that  
\[ \| (\cN+1)^{-1/2} \xi_N \| \leq C \]
uniformly in $N$. Moreover, we have
\[ \xi_N = \sum_{\ell=0}^\infty \xi_N^{(\ell)} a^* (\ph)^\ell \Omega \]
with the coefficients
\[ \xi_N^{(\ell)} = \sum_{j=0}^\ell (-1)^j N^{j-\ell/2} \frac{N!}{(N-\ell+j)! (\ell-j)! j!} \]
Notice that the coefficients $\xi_N^{(\ell)}$ satisfy the recursion
\[ \xi_N^{(\ell)} = \frac{1-\ell}{\ell} N^{-1/2} \xi_N^{(\ell-1)} - \frac{1}{\ell} \xi_N^{(\ell-2)} \]
with $\xi_N^{(0)} = 1$ and $\xi_N^{(1)} = 0$. 
\end{lem}


\begin{thebibliography}{77777} \label{sec:bib}

\bibitem{BKS}
G. Ben Arous, K. Kirkpatrick, B. Schlein. A central limit theorem in many--body quantum dynamics.  {\it Comm. Math. Phys.} {\bf 321} (2013), no. 2, 371--417.

\bibitem{BDS}
N. Benedikter, G. de Oliveira, B. Schlein. Quantitative derivation of the Gross-Pitaevskii equation. Preprint arXiv:1208.0373.

\bibitem{C}
X. Chen. Second order corrections to mean field evolution for weakly interacting bosons in the case of 3-body interactions. {\it Arch. Ration. Mech. Anal.} {\bf 203} (2012), no.2 , 455-497. 

\bibitem{CLS}
L. Chen, J. Oon Lee, B. Schlein. Rate of convergence towards Hartree dynamics. {\it J. Stat. Phys.} {\bf 144} (2011), no. 4, 872--903.

\bibitem{CE} M. Cramer, J. Eisert. A quantum central limit theorem for non-equilibrium systems: exact local relaxation of correlated states, {\it New J. Phys.} {\bf 12}, 055020 (2009).

\bibitem{CH} C.D. Cushen, R.L. Hudson. A quantum-mechanical central limit theorem. 
{\it J. Appl. Prob.} {\bf 8} (1971), 454.

\bibitem{ESY1} 
L. Erd{\H{o}}s, B. Schlein, H.-T. Yau. 
Derivation of the cubic nonlinear Schr\"odinger equation from
quantum dynamics of many-body systems. {\it Invent. Math.} {\bf 167} (2007), 515-614.

\bibitem{ESY2} 
L. Erd{\H{o}}s, B. Schlein, H.-T. Yau. Derivation of the Gross-Pitaevskii equation for the dynamics of Bose-Einstein condensate. {\it Ann. Math. (2)} {\bf 172} (2010), 291-370. 

\bibitem{ESY3} 
L. Erd{\H{o}}s, B. Schlein, H.-T. Yau. Rigorous derivation of the Gross-Pitaevskii equation with a large interaction potential. {\it J. Amer. Math. Soc.} {\bf 22} (2009), no. 4, 1099-1156.  

\bibitem{EY}
L. Erd\H os, H.--T. Yau. Derivation for the nonlinear Schr\"odingier equation from a many body Coulomb system. {\em Adv. Theor. Math. Phys.} {\bf 5} (2001), no. 6, 1169--1205. 

\bibitem{FKS}
J. Fr\"ohlich, A. Knowles, S. Schwarz. On the mean-field limit of bosons with Coulomb two-body interaction. {\it Comm. Math. Phys.} {\bf 288} (2009), no. 3, 1023--1059. 

\bibitem{GS}
Grech, R. Seiringer. The Excitation Spectrum for Weakly Interacting Bosons in a Trap. {\it Comm. Math. Phys.} {\bf 322} (2013), no. 2, 559--591.

\bibitem{GMM}
M. Grillakis, M. Machedon, D. Margetis. Second-order corrections to mean field evolution of weakly interacting bosons. I. {\it Comm. Math. Phys.} {\bf 294} (2010), no. 1, 273--301.

\bibitem{GMM2}
M. Grillakis, M. Machedon, D. Margetis. Second-order corrections to mean field evolution of weakly interacting bosons. II. {\it Adv. Math.} {\bf 228} (2011), no. 3, 1788-1815.

\bibitem{GV} 
J. Ginibre, G. Velo. The classical
field limit of scattering theory for non-relativistic many-boson
systems. I and II. \textit{Commun. Math. Phys.} \textbf{66} (1979),
37--76, and \textbf{68} (1979), 45--68.

\bibitem{GVV} D. Goderis, A. Verbeure, P. Vets. About the mathematical theory of quantum fluctuations. In {\it Mathematical Methods in Statistical Mechanics}. Leuven Notes in Mathematical and Theoretical Physics. Series A: Mathematical Physics, {\bf 1}. Leuven University Press, Leuven (1989).

\bibitem{Ha} M. Hayashi. Quantum estimation and the quantum central limit theorem. {\it Science And Technology} {\bf 227} (2006), 95.  

\bibitem{H}
K. Hepp. The classical limit for quantum mechanics correlation functions. {\em Comm. Mth. Phys.} {\bf 35} (1974), 265--277.

\bibitem{HL} K. Hepp, E. H. Lieb. Phase transitions in reservoir-driven open systems with applications to lasers and superconductors. {\it Helv. Phys. Acta} {\bf 46} (1973), 573.

\bibitem{JPP} V. Jak\v{s}i{\'c}, Y. Pautrat, C.-A. Pillet. A quantum central limit theorem for sums of iid random variables. {\it J. Math. Phys.} {\bf 51} (2010), 015208.

\bibitem{KP} A. Knowles, P. Pickl. Mean-field dynamics: singular potentials and rate of convergence. {\it Comm. Math. Phys.} {\bf 298} (2010), no. 1, 101-138.

\bibitem{Ku} G. Kuperberg. A tracial quantum central limit theorem. {\it Trans. Amer. Math. Soc.} {\bf 357} (2005), 549.

\bibitem{LNS}
M. Lewin, Nam, B. Schlein. Fluctuations around Hartree states in the mean-field regime. 
Preprint arXiv:1307.0665.

\bibitem{LNSS}
M. Lewin, P. T. Nam, S. Serfaty, J. P. Solovej. Bogoliubov spectrum of interacting {B}ose gases. Preprint arXiv:1211.2778.

\bibitem{RS}
I. Rodianski, B. Schlein. Quantum fluctuations and rate of convergence towards mean filed dynamics. {\em Comm. Math. Phys} {\bf 291} (2009), no. 1, 31--61.

\bibitem{S} 
R. Speicher. A noncommutative central limit theorem. {\it Math. Z.} {\bf 209} (1992), no. 1, 55--66.

\bibitem{Sp}
H. Spohn. Kinetic equations from Hamiltonian dynamics. {\em Rev. Mod. Phys.} {\bf 52} (1980), no. 3, 569--615. 

\end{thebibliography}
\end{document}